\documentclass[journal]{IEEEtran}

\IEEEoverridecommandlockouts                              

\usepackage{changes}
\usepackage{authblk}
\usepackage{amsfonts}
\usepackage{amsthm}
\usepackage{amsmath,amssymb}
\usepackage{bm}
\usepackage{hyperref,cite}
\usepackage{cleveref}
\usepackage{color}
\usepackage{algorithm}
\usepackage{algorithmic}
\let\labelindent\relax
\usepackage{enumitem}
\usepackage{float}
\usepackage{graphicx}
\usepackage{subfigure}
\usepackage{booktabs}
\usepackage{threeparttable}
\usepackage{hhline}
\usepackage{scalerel}
\usepackage{ulem}
\usepackage{tikz}

\newtheorem{theorem}{Theorem}

\newtheorem{remark}{Remark}
\newtheorem{lemma}{Lemma}

\newtheorem{definition}{Definition}
\DeclareMathOperator*{\argmin}{arg\,min} 
\newcommand{\diag}{\mathop{\rm diag}\nolimits}

\newcommand{\rr}{{\mathbb R}}
\newcommand{\ba}[1]{\begin{array}{#1}}
\newcommand{\ea}{\end{array}}

\newcommand*{\pdot}{\mathbin{\scalerel*{\boldsymbol\odot}{\circ}}}

\newcommand{\tikzxmark}{%
\tikz[scale=0.23] {
    \draw[line width=0.7,line cap=round] (0,0) to [bend left=6] (1,1);
    \draw[line width=0.7,line cap=round] (0.2,0.95) to [bend right=3] (0.8,0.05);
}}
\newcommand{\tikzcmark}{%
\tikz[scale=0.23] {
    \draw[line width=0.7,line cap=round] (0.25,0) to [bend left=10] (1,1);
    \draw[line width=0.8,line cap=round] (0,0.35) to [bend right=1] (0.23,0);
}}

\begin{document}
\title{\LARGE \bf
EIQP: Execution-time-certified and Infeasibility-detecting QP Solver}
\author{Liang Wu$^{1}$, Wei Xiao$^{1}$, Richard D. Braatz$^{1}$, Fellow, IEEE%
\thanks{$^{1}$Massachusetts Institute of Technology, Cambridge, MA 02139, USA, {\tt\small \{liangwu,weixy,braatz\}@mit.edu}\\
}
}

\maketitle

\begin{abstract}
Solving real-time quadratic programming (QP) is a ubiquitous task in control engineering, such as in model predictive control and control barrier function-based QP. In such real-time scenarios, certifying that the employed QP algorithm can either return a solution within a predefined level of optimality or detect QP infeasibility before the predefined sampling time is a pressing requirement. This article considers convex QP (including linear programming) and adopts its homogeneous formulation to achieve infeasibility detection. Exploiting this homogeneous formulation, this article proposes a novel infeasible interior-point method (IPM) algorithm with the best theoretical $O(\sqrt{n})$ iteration complexity that feasible IPM algorithms enjoy. The iteration complexity is proved to be \textit{exact} (rather than an upper bound), \textit{simple to calculate}, and \textit{data independent}, with the value $\left\lceil\frac{\log(\frac{n+1}{\epsilon})}{-\log(1-\frac{0.414213}{\sqrt{n+1}})}\right\rceil$ (where $n$ and $\epsilon$ denote the number of constraints and the predefined optimality level, respectively), making it appealing to certify the execution time of online time-varying convex QPs. The proposed algorithm is simple to implement without requiring a line search procedure (uses the full Newton step), and its C-code implementation (offering MATLAB, Julia, and Python interfaces) and numerical examples are publicly available at \url{https://github.com/liangwu2019/EIQP}.
\end{abstract}


\section{Introduction}
Like solving a linear system, solving a quadratic programming (QP) problem has also been ubiquitous and extensively used including in finance, machine learning, operations, computer vision, energy, transportation, bioinformatics, signal processing, robotics, and control. For example, \textit{i)} model predictive control (MPC), which formulates a receding horizon control problem as a QP that is solved at each sampling time, has been widely spread in the industry for controlling multivariable systems subject to constraints \cite{qin2003survey} and \textit{ii)} CLF-CBF-QP \cite{ames2016control}, a control framework that combines stability (via control Lyapunov function) and safety (via control barrier functions) into a real-time QP, has been widely used in autonomous systems. Specifically regarding applications, \textit{i)} the quadrotor maneuvering field has a well-known QP-based framework \cite{mellinger2011minimum} for generating minimum snap trajectories with smooth transitions through waypoints while satisfying constraints on velocities, accelerations, and inputs; \textit{ii)} the robotics field often uses QP as a unified framework for task-space control paradigms \cite{escande2014hierarchical,bouyarmane2018quadratic}. All these QPs fall into the categories of \textit{online}, as they are required to be executed in real time, in contrast to the \textit{offline} QPs, which do not have this requirement.

Solving \textit{real-time} QPs poses a pressing and open challenge for algorithms: \textit{how to guarantee that the QP algorithm can return the optimal solution or detect infeasibility before the predefined sampling (or feedback) time}. This is called the \textit{execution time certificate} problem. Without the \textit{execution time certificate}, control signals (the backup control signal if detecting QP infeasibility) fail to return on time, which renders the system open-loop and potentially leads to safety concerns, especially in safety-critical process systems.

Certifying the execution time of QP algorithms can be reduced to certifying the number of iterations of QP algorithms if each iteration performs a fixed number of floating-point operations ([flops]). In fact, nearly all QP solvers have default settings that specify a maximum number of iterations, seemingly addressing the worst-case \textit{execution time certificate} problem. However, this is a \textit{cheating} approach for addressing \textit{execution time certificate} problem, because it can not guarantee that all returned solutions have \textit{the same level of optimality}.

An execution-time-certified QP algorithm should precisely be defined as one that \textit{either returns an optimality-certified solution (within a predefined level of optimality) or detects QP infeasibility before the predefined sampling time}. First of all, the stability and safety guarantee of MPC and CLF-CBF-QP control frameworks both assume that the QP at each sampling time is exactly solved, in other words, the QP solution should be \textit{optimality-certified}. Second, the \textit{optimality-certified} QP solution is also pursued in practice. When the plant suffers a sudden disturbance, MPC often requires that the QP algorithm runs more iterations at this time instant to ensure good performance. For example, in the case of an autonomous car or drone making a sharp turn, if the MPC-based QP does not run enough iterations, an autonomous car or drone may fail and deviate from the desired trajectory. Recently, the \textit{execution time certificate} problem has attracted significant scholarly interest and remains a vibrant research area, especially within the context of MPC \cite{richter2011computational, giselsson2012execution, patrinos2013accelerated, cimini2017exact, arnstrom2019exact, cimini2019complexity,arnstrom2020complexity, arnstrom2021unifying, okawa2021linear, wu2023direct, wu2024time, wu2024execution,wu2024parallel}.

\subsection{Related work}
Current QP algorithms tailored for efficiently solving MPC include first-order, active set, and interior point methods \cite{stellato2020osqp, wu2023simple, wu2023construction,ferreau2014qpoases,wang2009fast}. The iteration complexity analysis of first-order methods is provided in \cite{richter2011computational, giselsson2012execution, patrinos2013accelerated} for linear MPC problems, but their iteration complexity result is dependent on the QP problem data such as the Hessian matrix. Their \textit{data-dependent} iteration complexity result cannot guarantee the time-invariant number of iterations for MPC-based QP with time-varying data such as  Real-Time Iteration (RTI)-based nonlinear MPC \cite{gros2020linear}.

Although active set methods often have good performance in small- to medium-scale problems in practice, in theory, they can have an exponential number of iterations in the worst case \cite{klee1972good}. Refs.\  \cite{cimini2017exact, arnstrom2019exact, cimini2019complexity,arnstrom2020complexity, arnstrom2021unifying} develop a computationally complicated and expensive (thus offline) worst-case partial enumeration technique, to certify the worst-case number of iterations for active set methods. Their iteration complexity results are also \textit{data-dependent} and cannot be used in QPs that have time-varying data. This also happens in work \cite{okawa2021linear}, which proposes an $N$-step algorithm for input-constrained MPC problems, and its worst-case iteration is the problem dimension $N$, but assuming that a \textit{modified N-step vector} is given. The authors in \cite{okawa2021linear} propose solving a linear program (LP) to find this \textit{modified N-step vector}, dependent on QP data, making it not suitable for real-time QP that have time-varying data.

Interior point methods (IPMs) can be categorized as being \textit{practical} (such as Mehrotra predictor-corrector IPMs \cite{mehrotra1992implementation}) or \textit{theoretical}. Thanks to the super-fast convergence speed ($O(\log(n))$ iteration complexity, with likely $<50$ iterations), the Mehrotra predictor-corrector IPMs are the foundation for most interior point software such as \cite{zanelli2020forces}, but are heuristic and may diverge for some examples \cite[see p.\ 411]{nocedal2006numerical}, \cite{cartis2009some}, without certified global convergence proof. \textit{Theoretical} IPMs can be classified as being infeasible or feasible, which generally have certified $O(n)$ and $O(\sqrt{n})$ iteration complexity \cite{wright1997primal,ye2011interior}, respectively. However, the faster iteration complexity of feasible IPMs comes at the cost of requiring a strictly feasible initial point lying in the neighborhood of the central path. Finding this strictly feasible initial point for a general QP requires solving an LP. To avoid this limit, our previous work \cite{wu2023direct} for the first time proposed a \textit{cost-free} initialization strategy for feasible IPMs by exploiting the structure of the considered box-constrained QP (arising from input-constrained MPC). One interesting feature of our algorithm \cite{wu2023direct} is that its iteration complexity is the \textit{exact} value: $
\!\left\lceil\frac{\log(\frac{2n}{\epsilon})}{-2\log(\frac{\sqrt{2n}}{\sqrt{2n}+\sqrt{2}-1})}\right\rceil \! + 1
$, rather than an upper bound as provided in previous works. This \textit{exact} iteration complexity is \textit{data independent} (only \textit{dimension dependent}), so our algorithm can certify the iteration of QP with time-varying data \cite{wu2024execution}. Although our algorithm \cite{wu2023direct} is only valid for box-constrained QP, Ref.\ \cite{wu2024parallel} then extends it to general strictly convex QPs via an $\ell_1$-penalty soft constraint which is transformed to box-constrained QP via duality theory. This soft constraint approach can offer a viable solution when the strictly convex QP is infeasible, such as in cases where an MPC encounters unknown disturbances or modeling errors that render the MPC-based QP problem infeasible.

However, some scenarios require real-time QP solvers to have the ability to detect infeasibility. For example, the soft-constrained solution is always approximate and hard to guarantee the stability of MPC. In this case, the backup control law can be active when detecting the potential infeasibility of real-time QP. Naturally, detecting infeasibility should also be completed within the certified execution time; all the aforementioned iteration complexity works \cite{richter2011computational, giselsson2012execution, patrinos2013accelerated, cimini2017exact, arnstrom2019exact, cimini2019complexity,arnstrom2020complexity, arnstrom2021unifying, okawa2021linear, wu2023direct, wu2024time, wu2024execution} do not have infeasibility-detection capability and either assume the QP is feasible or restrict their consideration to box-constrained QPs. 

Regarding infeasibility detection in QP, \textit{i)} the first-order method, specifically the alternating direction method of multipliers, has been reported to possess the infeasibility-detection capability in \cite{banjac2019infeasibility} and its corresponding solver OSQP \cite{stellato2020osqp}, but lacks the iteration complexity analysis; \textit{ii)} the active set method, combined with the nonnegative least-squares (NNLS) formulation, has infeasibility-detection capability in \cite{bemporad2015quadratic} and can be combined with the offline (expensive) iteration complexity certification procedure \cite{ arnstrom2020exact} to determine the worst-case iterative behavior, but this NNLS formulation is only limited to strictly convex QPs; \textit{iii)} IPM, combined with the homogeneous slacked QP formulation \cite{raghunathan2021homogeneous}, possesses the infeasibility-detection capability but is limited to strictly convex QPs and lacks a certified iteration complexity for the employed IPM algorithm. IPMs based on the homogeneous self-dual (HSD) formulation \cite{ye1994nl, xu1996simplified} are well known to be capable of detecting the infeasibility of LPs. Furthermore, Ref.\ \cite{ye1997homogeneous} extended the HSD formulation to the linear complementarity problem (LCP) (more general than LP), but their proposed IPM has \textit{data-dependent} iteration complexity $O(\sqrt{n}L)$ (where $L$ is the binary data length of LCP problem data). Ref.\ \cite{andersen1999homogeneous} proposed a homogeneous formulation for nonlinear monotone complementarity problems to support infeasibility detection, and provided iteration complexity analysis for the employed infeasible IPM, but its iteration complexity result is too conservative and  \textit{data dependent}, specifically depending on the scaled Lipschitz constants of the nonlinear monotone mapping. 

\subsection{Contributions}
This article extends the homogeneous formulation from \cite{andersen1999homogeneous} to detect the infeasibility of convex QPs and proposes a novel infeasible IPM algorithm to solve convex QPs. Our novel contributions are the following:
\begin{itemize}
    \item[1)] proves that this homogeneous formulation enables our proposed infeasible IPM algorithm to achieve the best theoretical $O(\sqrt{n})$ iteration complexity, that feasible IPM algorithms enjoy, so there is no need to find a strictly feasible initial point.
    
    \item[2)] proves that the iteration complexity of our algorithm is \textit{exact} (rather than the upper bound), \textit{simple to calculate}, and \textit{data independent}, with the value
    \[
\left\lceil\frac{\log(\frac{n+1}{\epsilon})}{-\log(1-\frac{0.414213}{\sqrt{n+1}})}\right\rceil
    \]
    (where $n$ and $\epsilon$ denote the number of constraints and the predefined optimality level, respectively), making it appealing to certify the execution time of online time-varying convex QPs (including LPs).
    \item[3)] shows that our algorithm is \textit{simple} to implement as it does not require a line-search procedure. Its C-code implementation is just one single C file, making it easy to integrate with other software and appealing when deployed in embedded production systems. Its MATLAB, Julia, and Python interfaces and numerical examples are publicly available at \url{https://github.com/liangwu2019/EIQP}. Moreover, our algorithm is \textit{parameter free} as there are no algorithm parameters, avoiding the tuning work with algorithm parameters.
\end{itemize}

To the best of the authors' knowledge, this article for the first time provides a convex QP solver with \textit{exact} and \textit{data-independent} \textbf{E}\textit{xecution time certificate} and \textbf{I}nfeasibility detection, thus referred to as EIQP.

\subsection{Notations}
$\left\lceil x\right\rceil$ maps $x$ to the least integer greater than or equal to $x$. $\mathbb{R}^n$ denotes the space of $n$-dimensional real vectors, $\mathbb{R}^n_{+}$ and $\mathbb{R}^n_{++}$ are the set of all non-negative and positive vectors of $\mathbb{R}^n$, respectively. $Q \succ 0$ ($Q\succeq 0$) denotes positive definiteness (semi-definiteness) of a square matrix $Q$, $Q^{\top}$ (or $z^{\top}$) denotes the transpose of matrix $Q$ (or vector $z$). $\min(z,y)$ denotes the minimum value between scalar $z$ and $y$. Given two arbitrary vectors $z,y \in\mathbb{R}^n_{++}$, their Hadamard product is $z\pdot y = (z_1y_1,z_2y_2,\cdots{},z_ny_n)^\top$, $\frac{z}{y}=z\pdot \big(y^{-1}\big)=\big(\frac{z_1}{y_1},\frac{z_2}{y_2},\cdots{},\frac{z_n}{y_n}
\big)^{\!\top}$. For $z,y\in\mathbb{R}^n$, let $\mathrm{col}(z,y)=[z^{\top},y^{\top}]^{\top}$ and $\max(z,y)=(\max(z_1,y_1),\max(z_2,y_2),\cdots{},\max(z_n,y_n))^\top$. For a vector $z\in\mathbb{R}^n$, its Euclidean norm is $\|z\|=\sqrt{z_1^2+z_2^2+\cdots+z_n^2}$, $\|z\|_1=\sum_{i=1}^{n}|z_i|$, $\|z\|_{\infty}=\max_i |z_i|$, $\diag(z):\mathbb{R}^n\rightarrow\mathbb{R}^{n\times n}$ maps an vector $z$ to its corresponding diagonal matrix, $\sqrt{z}=\left(\sqrt{z_1}, \sqrt{z_2},\cdots{},\sqrt{z_n}\right)^{\!\top}$ and $\frac{1}{z}=(\frac{1}{z_1},\frac{1}{z_2},\cdots{},\frac{1}{z_n})^{\top}$.

\section{Problem formulations}
This article considers the convex QP,
\begin{equation}\label{eqn_QP}
    \begin{aligned}
        \min_z&~ \frac{1}{2}z^\top Q(t) z + c(t)^\top z \\
        \mathrm{s.t.}&~ A(t)z\geq b(t),\\
        &~z\geq 0
    \end{aligned}
\end{equation}
where the data $Q(t)=Q(t)^\top\in\rr^{n_z\times n_z}\succeq0$, $c(t)\in\rr^{n_z}$, $A(t)\in\rr^{n_b\times n_z}$, $b(t)\in\rr^{n_b}$, are all often time-varying in real-time applications. Problem\ \eqref{eqn_QP} is common such as in real-time MPC with state and input constraints and CLF-CBF-QP with input bounds, see subsections \ref{subsection_MPC} and \ref{subsection_CBF_QP}. 
For example, a typical MPC or CLF-CBF-QP problem, $\min_z \frac{1}{2}z^\top Q(t)z+c(t)^\top z,\text{s.t.~} lb(t)\leq z\leq ub(t), A(t)z\leq b(t)$, (even its soft-constrained formulation), can be easily formulated as Problem \eqref{eqn_QP} by letting $z\leftarrow z-lb(t)$, $c(t)\leftarrow c(t)+Q(t)lb(t)$, $A(t)\leftarrow[-I;-A(t)]$, and $b(t)\leftarrow[lb(t)-ub(t);A(t)lb(t)-b(t)]$. 
A standard QP formulation can be reformulated as \eqref{eqn_QP}.
\begin{remark}
Consider the standard QP formulation,
\begin{equation}\label{eqn_standard_QP}
    \begin{aligned}
        \min_z&~ \frac{1}{2}z^\top Q(t) z + c(t)^\top z \\
        \mathrm{s.t.}&~ A(t)z\geq b(t)
\end{aligned}
\end{equation}
which does not have the additional non-negative constraints $z\geq0$. It is trivial to transform this QP to  \eqref{eqn_QP} by introducing $\Bar{z}\triangleq \mathrm{col}(z^+,z^-)$, where $z^+,z^-$ denote the positive and negative part of $z$, respectively. That is, $z=z^+ - z^-$ and $z^+_i z^-_i=0, i=1,\cdots{},n_z$. Based on this representation, a standard QP formulation can be reformulated as
\begin{equation}\label{eqn_transformed_QP}
\begin{aligned}
    \min_{\Bar{z}}&~\frac{1}{2}\Bar{z}^\top\Bar{Q}(t)\Bar{z}+\Bar{c}(t)^\top\Bar{z}\\
    \mathrm{s.t.}&~\Bar{A}(t)\Bar{z}\geq b(t)\\
     &~\Bar{z}\geq 0
\end{aligned}    
\end{equation}
where $\Bar{A}(t)=\left[A(t),-A(t)\right]\!,$
\[ 
\Bar{Q}(t)=\left[\begin{array}{cc}
    Q(t) & -Q(t) \\
     -Q(t) & Q(t)
\end{array}\right]\!, \text{ and }\Bar{c}(t)=\left[\begin{array}{c}
     c(t)  \\
     -c(t) 
\end{array}\right].
\]
\end{remark}
Note that the constraint  $z^+_i z^-_i=0, i=1,\cdots{},n_z$ is eliminated since the following Lemma holds,

\begin{lemma}
    Let $\Bar{z}=\mathrm{col}(z^+,z^-)$ be the minimizer of \eqref{eqn_transformed_QP}, then it always satisfies $z^+_i z^-_i=0, i=1,\cdots{},n_z$.
\end{lemma} 
\begin{proof}
    The proof is by contradiction. Assume that the minimizer $\mathrm{col}(z^+,z^-)\in\argmin$ \eqref{eqn_transformed_QP} does not satisfy $z^+_iz^-_i=0~\text{for}~i=1,\cdots{},n_z$. Then another $\mathrm{col}(\Tilde{z}^+,\Tilde{z}^-)$ can be constructed as
    \[
    \begin{aligned}
    &\Tilde{z}^+_i=z^+_i-\min(z^+_i,z^-_i),~\text{for}~i=1,\cdots{},n_z\\
        &\Tilde{z}^-_i=z^-_i-\min(z^+_i,z^-_i),~\text{for}~i=1,\cdots{},n_z
    \end{aligned}
    \]
    Clearly, $\mathrm(\Tilde{z}^+,\Tilde{z}^-)\geq0$ and $\Tilde{z}^+_i\Tilde{z}^-_i=0~\text{for}~i=1,\cdots{},n_z$
    hold. Moreover,
    \[
        \Tilde{z}^+-\Tilde{z}^- = z^+-z^-,
    \]
    thus the inequality constraints $\Bar{A}(t)\mathrm{col}(\Tilde{z}^+,\Tilde{z}^-)\geq b(t)$ hold. Then, the objective function of \eqref{eqn_transformed_QP} at $\mathrm{col}(\Tilde{z}^+,\Tilde{z}^-)$ and $\mathrm{col}(z^+,z^-)$ is equal, which contradicts the assumption that $\mathrm{col}(z^+,z^-)\in\argmin$ \eqref{eqn_transformed_QP}, which completes the proof.
\end{proof}

\begin{remark}
    When $Q=0$, Problem \eqref{eqn_standard_QP} becomes an LP. It is emphasized that the proposed Algorithm \ref{alg_efficient} (see Section \ref{sec_alg_2}), with execution time certificate and infeasibility detection, is also valid for solving LP.
\end{remark}

\subsection{KKT condition and LCP formulation}
The Karush–Kuhn–Tucker (KKT) condition \cite[Ch.\ 5]{boyd2004convex} for \eqref{eqn_QP} is
\begin{equation}\label{eqn_KKT}
    \begin{aligned}
        &v = Q(t)z-A(t)^\top y + c(t)\\
        &w = A(t)z-b(t)\\
        &v\pdot z=0, v\geq0, z\geq0\\
        &w\pdot y=0, w\geq0, y\geq0
    \end{aligned}
\end{equation}
where $y\in\rr^{n_b}$ denotes the Lagrangian variable for inequality constraints $A(t)z\geq b(t)$, $v\in\rr^{n_z}$ denotes the Lagrangian variable for $z\geq0$, and $w\triangleq A(t)z-b(t)\in\rr^{n_b}$ is the slack variable. Eq.\ \eqref{eqn_KKT} is an LCP; for simplicity, we adopt the notations,
\[
x\triangleq\mathrm{col}(z,y)\in\rr^n,~s\triangleq\mathrm{col}(v,w)\in\rr^n,
\]
where the dimension $n=n_z+n_b$. Then Eq. \eqref{eqn_KKT} is reformulated as
\begin{equation}\label{eqn_LCP}
    \begin{aligned}
    \text{LCP}~\left\{\begin{array}{l}
        s = f(x)\triangleq Mx+q\\
        x\pdot s=0\\
        x\geq0, s\geq0        
    \end{array}
        \right.
    \end{aligned}
\end{equation}
where 
\[
M \triangleq\left[\begin{array}{cc}
            Q(t) &  -A(t)^\top\\
            A(t) & 0
        \end{array}\right]\!,\quad q\triangleq \left[\begin{array}{c}
             c(t)  \\
             -b(t)
        \end{array}\right].
\]
Note that $M\neq M^\top$. $M\succeq0$, and $M+M^\top\succeq0$ (easy to prove as $Q(t)\succeq0$).

The QP \eqref{eqn_QP} may be infeasible, such that the set $\{z\in\rr^{n_z}: A(t)z\geq b(t),z\geq0\}$ is empty and there is no solution for the LCP \eqref{eqn_LCP}. To achieve infeasibility detection, this paper adopts a homogeneous LCP formulation from \cite{andersen1999homogeneous} and LCP \eqref{eqn_LCP} is a linear case of nonlinear monotone complementarity problem \cite{andersen1999homogeneous} by the following Lemma.
\begin{lemma}\label{lemma_f_is_monotone}
$f(x)$ is a continuous monotone mapping in $\mathbb{R}^{n}_{+}$, namely $\forall x_1,x_2\in\mathbb{R}^{n}_{+}$, $(x_1-x_2)^\top(f(x_1)-f(x_2))\geq0$.
\end{lemma}
\begin{proof}
     For two arbitrary points $x_1,x_2\in\mathbb{R}^{n}_{+}$, define the function 
    \[
    \phi(\delta)\triangleq (x_2-x_1)^\top(f(x_\delta)-f(x_1))
    \]
    with $\delta\in[0,1]$ and $x_{\delta}\triangleq(1-\delta)x_1+\delta x_2\in\mathbb{R}^{n}_{++}$. Clearly, $\phi(0)=0$. As $\nabla f(x)=M$ is positive semi-definite in $\mathbb{R}^{n}_{++}$,
\[
\phi^\prime(\delta) =(x_2-x_1)^\top M(x_2-x_1)\geq0,
\]
so $\phi(\delta)$ is non-decreasing. Hence $\phi(1)\geq0$ is proved, which completes the proof.
\end{proof}

\subsection{Homogeneous LCP formulation}
By introducing two additional scalars $\tau\in\mathbb{R}$ and $\kappa\in\mathbb{R}$, a homogeneous LCP (HLCP) is formulated as 
\begin{equation}\label{eqn_HLCP}
    \begin{aligned}
        \text{HLCP}~\left\{
        \begin{array}{c}
            \left[\begin{array}{l}
             s  \\
              \kappa 
        \end{array}\right]
         = \psi(x,\tau)   \\
          \left[\begin{array}{l}
             x  \\
              \tau 
        \end{array}\right] \pdot  \left[\begin{array}{l}
             s  \\
              \kappa 
        \end{array}\right] = 0\\
        (x,\tau, s,\kappa)\geq0
        \end{array}
        \right.
    \end{aligned}
\end{equation}
where 
\begin{equation}\label{eqn_psi}
    \psi(x,\tau) \triangleq \!\left[\begin{array}{c}
             Mx+q\tau  \\
             -x^\top Mx/\tau-x^\top q 
        \end{array}\right]: \,\mathbb{R}^{n+1}_{++} \rightarrow \mathbb{R}^{n+1}.
\end{equation}
Note that the HLCP \eqref{eqn_HLCP} inherits the monotone feature by the following Lemma.
\begin{lemma}\label{lemma_psd}
    The gradient of $\psi(x,\tau)$ is
    \begin{equation}\label{eqn_nabla_psi}
        \nabla \psi(x,\tau) = \left[\begin{array}{cc}
        M &  q\\
        -x^\top( M+M^\top)/\tau-q^\top &  x^\top Mx/\tau^2
    \end{array}\right]
    \end{equation}
    and $\nabla \psi(x,\tau)$ is semi-positive definite in $\mathbb{R}^{n+1}_{++}$.

    Furthermore,  $\psi(x,\tau)$ is a continuous monotone mapping in $\mathbb{R}^{n+1}_{++}$.
\end{lemma}
\begin{proof}
    By the definition of $\psi(x,\tau)$ \eqref{eqn_psi}, its gradient $ \nabla \psi(x,\tau)$ indeed follows \eqref{eqn_nabla_psi}. 
    
    Given an arbitrary $\mathrm{col}(d_x,d_\tau)\in\mathbb{R}^{n+1}$ (where $d_x\in\mathbb{R}^n$, $d_\tau\in\mathbb{R}$) and $M$ is semi-positive definite, 
\begin{equation}\label{eqn_dx_d_tau}
    \begin{aligned}
    &~\left[d_x^\top, d_\tau \right]\nabla \psi(x,\tau)\left[\begin{array}{c}
    d_x \\
    d_\tau
    \end{array}\right] \\
    &= d_x^\top Md_x + (d_x^\top q)d_\tau-\frac{d_\tau}{\tau} x^\top (M+M^\top)d_x-d_\tau q^\top d_x \\
    &\quad + \frac{d_\tau}{\tau} x^\top Mx\frac{d_\tau}{\tau} \\
    &=\frac{1}{2}d_x^\top(M+M^\top)d_x-\frac{d_\tau}{\tau} x^\top (M+M^\top)d_x\\
    &\quad+\frac{1}{2}\frac{d_\tau}{\tau} x^\top (M+M^\top)x\frac{d_\tau}{\tau}\\
    &=\frac{1}{2}\left(d_x-\frac{d_\tau}{\tau}x\right)^{\!\!\top} (M+M^\top)\left(d_x-\frac{d_\tau}{\tau}x\right) \geq 0,
    \end{aligned}
    \end{equation}
    which completes the first part of the proof.

    Similarly, like Lemma \ref{lemma_f_is_monotone} in proving the monotone property, $\psi(x,\tau)$ is a continuous monotone mapping in $\mathbb{R}^{n+1}_{++}$ under the fact that its gradient $\nabla \psi(x,\tau)$ is semi-positive definite in $\mathbb{R}^{n+1}_{++}$.
\end{proof}

\begin{lemma}
 HLCP \eqref{eqn_HLCP} is always asymptotically feasible. Every asymptotically feasible solution of HLCP \eqref{eqn_HLCP} is an asymptotically ``optimal" or ``complementary" solution.
\end{lemma}
\begin{proof}
    HLCP \eqref{eqn_HLCP} is said to be asymptotically feasible if and only if there is positive and bounded iterates $(x^k,\tau^k,s^k,\kappa^k)>0, k=1,2,\cdots{}$ such that $\lim_{k\rightarrow\infty}\mathrm{col}(s^k,\kappa^k)-\psi(x^k,\tau^k)\rightarrow0$. An asymptotically feasible solution $(x^k,\tau^k,s^k,\kappa^k)$ of HLCP \eqref{eqn_HLCP} such that $(x^k)^\top s^k + \tau^k \kappa^k=0$
    is said to be an asymptotically ``optimal" or ``complementary" solution of HLCP \eqref{eqn_HLCP}.
    
    Take $x^k = (\frac{1}{2})^ke, \tau^k = (\frac{1}{2})^ke, s^k = (\frac{1}{2})^ke$, and $\kappa^k = (\frac{1}{2})^k$. Then, as $k\rightarrow\infty$,
    \[
    \begin{aligned}
        \left[\begin{array}{c}
             s^k  \\
             \kappa^k
        \end{array}\right]-\psi(x^k,\tau^k) =  \left(\frac{1}{2}\right)^{\!\!k}\left[\begin{array}{c}
            e-Me-q  \\
             1+e^\top M e+e^\top q
        \end{array}\right]\!\rightarrow0,
    \end{aligned}
    \]
    which completes the first part of the proof.

    For each $(x^k,\tau^k,s^k,\kappa^k)$ asymptotically feasible solution of HLCP \eqref{eqn_HLCP}, by \eqref{eqn_psi} it always holds that $(x^k)^\top s^k + \tau^k \kappa^k=0$, which completes the second part of the proof.
\end{proof}

\begin{definition}
    An ``optimal" or ``complementary" solution $(x^*,s^*)$ for LCP \eqref{eqn_LCP} (or $(x^*,\tau^*,s^*,\kappa^*)$) for HLCP \eqref{eqn_HLCP}) is said to be a maximal complementary solution such that the number of positive components in $(x^*,s^*)$ (or in $(x^*,\tau^*,s^*,\kappa^*)$) is maximal.
\end{definition}
\begin{lemma}
    (see \cite[Theorem 2.3]{guler1993existence}): The indices for those positive components are invariant among all maximal complementary solutions for LCP \eqref{eqn_LCP} (or HLCP \eqref{eqn_HLCP}).
\end{lemma}
Thus, finding a maximal complementary solution for HLCP \eqref{eqn_HLCP} is equivalent to finding a maximal complementary solution for LCP \eqref{eqn_LCP},
and the relationship between the solutions of HLCP \eqref{eqn_HLCP} and LCP \eqref{eqn_LCP} can be described by the following Lemma.
\begin{lemma}\label{lemma_HLCP_solution}
(see \cite[Thm.\ 1]{andersen1999homogeneous}): Let $(x^*,\tau^*,s^*,\kappa^*)$ be a maximal complementary for HLCP \eqref{eqn_HLCP}. Then
    \begin{itemize}
        \item[i)] LCP \eqref{eqn_LCP} has a solution if and only if $\tau^*>0$. In this case, $\mathrm{col}(x^*/\tau^*,s^*/\tau^*)$ is a maximal complementary solution for LCP  \eqref{eqn_LCP};
        \item[ii)] LCP \eqref{eqn_LCP} is infeasible if and only if $\kappa^*>0$. In this case, $\mathrm{col}(x^*/\kappa^*,s^*/\kappa^*)$ is a certificate to prove infeasibility.
    \end{itemize}
\end{lemma}
Thus, finding the solution or detecting the infeasibility of LCP \eqref{eqn_LCP} is equivalent to finding a maximal complementary solution of HLCP \eqref{eqn_HLCP}.

\subsection{Central path of HLCP}
Denote $\Bar{r}\triangleq\mathrm{col}(r,\Tilde{r})\in\mathbb{R}^{n+1}$ (where $r\in\mathbb{R}^{n}, \Tilde{r}\in\mathbb{R}$) as the residual vector for the nonlinear equation 
\[
  \left[\begin{array}{c}
                s \\
                \kappa 
            \end{array}\right] -\psi(x,\tau) = 0
\]
from HLCP \eqref{eqn_HLCP}. For example, given an initial point $x^0>0, s^0>0, \tau^0>0, k^0>0$, the initial residual vector is shown as follows,
\[
\begin{aligned}
r^0 &= s^0 - Mx^0 - q\tau^0, \\
\Tilde{r}^0 &= \kappa^0 + (x^0)^\top Mx^0/\tau^0+(x^0)^\top q   
\end{aligned}
\]
Next, the following lemma defines the central path of HLCP \eqref{eqn_HLCP} and its existence.
\begin{lemma}\label{lemma_HLCP_central_path} (see \cite[ Thm.\ 2]{andersen1999homogeneous}):
Considering HLCP \eqref{eqn_HLCP}:
    \begin{itemize}
        \item[(i)] for any $0<\theta\leq1$, there exists a strictly positive point $(x>0,\tau>0,s>0,\kappa>0)$ such that
\begin{equation}\label{eqn_central_path_1}
\left[\begin{array}{c}
                s \\
                \kappa 
\end{array}\right] -\psi(x,\tau) = \theta\left[\begin{array}{c}
    r^0 \\
    \Tilde{r}^0 
\end{array}\right].
\end{equation}
    \item[(ii)] starting from $(x^0=e,\tau^0=1,s^0=e,\kappa^0=1)$, for any $0<\theta\leq1$ there is a unique strictly positive point $(x(\theta),\tau(\theta),s(\theta),k(\theta))$ satisfying \eqref{eqn_central_path_1} and 
\begin{equation}\label{eqn_central_path_2}
\left[\begin{array}{c}
                xs \\
                \tau \kappa 
\end{array}\right] = \theta e.
        \end{equation}
        \item[(iii)] for any $0<\theta\leq1$, the solution $(x(\theta),\tau(\theta),s(\theta),\kappa(\theta))$ in \textit{(ii)} is bounded. Therefore, the defined central path
\begin{equation}
C(\theta)\triangleq\left\{(x,\tau,s,\kappa): (\ref{eqn_central_path_1}),(\ref{eqn_central_path_2}),0<\theta\leq1\right\}
\end{equation}
        is a continuous bounded trajectory.
    \item[(iv)] when $\theta\rightarrow0$, any limit point $(x(\theta),\tau(\theta),s(\theta),\kappa(\theta))$ is a maximal complementary solution for HLCP \eqref{eqn_HLCP}.
    \end{itemize}    
\end{lemma}
By Lemma \ref{lemma_HLCP_central_path}, interior point algorithms can generate iterates with a neighborhood of $C(\theta)$ and converge toward a maximal complementary solution for HLCP \eqref{eqn_HLCP}.

\section{Infeasible full Newton IPM Algorithm for HLCP}
The previous section primarily explains how the homogeneous formulation, HLCP \eqref{eqn_HLCP}, enables the detection of infeasibility. Remarkably, this homogeneous formulation also surprisingly allows infeasible IPMs to achieve an iteration complexity of $O(\sqrt
n)$, a property traditionally reserved for feasible IPMs.

In this section, we introduce an infeasible IPM algorithm with the full Newton step, whose iteration complexity is proved to be \textit{data independent} (\textit{only dimension dependent}) and \textit{exact} with the value:
\[
\mathcal{N}=\left\lceil\frac{\log(\frac{n+1}{\epsilon})}{-\log(1-\frac{0.414213}{\sqrt{n+1}})}\right\rceil.
\]

We adopt the notations for simplicity,
\[
\Bar{x}\triangleq\mathrm{col}(x, \tau)\in\mathbb{R}^{n+1}, ~\Bar{s}\triangleq\mathrm{col}(s,k)\in\mathbb{R}^{n+1}.
\]
At the $k$th iterate $(\Bar{x}^k,\Bar{s}^k)>0$, the proposed algorithm adopts the full Newton direction, which is the solution of the linearized equations
\begin{subequations}\label{eqn_Newton_direction}
    \begin{align}
        d_{\Bar{s}} - \nabla\psi(\Bar{x}^k)d_{\Bar{x}} &= -\eta\Bar{r}^k\label{eqn_Newton_direction_1}\\
        \Bar{x}^k\pdot d_{\Bar{s}} + \Bar{s}^k\pdot d_{\Bar{x}} &=\gamma\Bar{\mu}^ke-\Bar{x}^k\pdot\Bar{s}^k\label{eqn_Newton_direction_2}
    \end{align}
\end{subequations}
where 
\[
\begin{aligned}
\Bar{r}^k &=\Bar{s}^k-\psi(\Bar{x}^k),\\
\quad\Bar{\mu}^k &=\frac{(\Bar{x}^k)^\top \Bar{s}^k}{n+1},   
\end{aligned}
\]
and $\eta, \gamma \in (0,1)$ are specified parameters by our algorithm for ensuring convergence (see the later analysis).

We first summarize our proposed infeasible full Newton IPM algorithm in Algorithm \ref{alg_fullNewton}, wherein the choice of $\eta, \gamma$ for convergence guarantee, the scaling and initialization strategy, and the exact number iterations are discussed in Subsections \ref{sec_convergence_analysis}, \ref{sec_initialization_stragegy} and \ref{sec_exact_iteration_complexity}, respectively.

\begin{algorithm}[H]
    \caption{Infeasible full Newton IPM algorithm for finding a maximal complementary solution of HLCP \eqref{eqn_HLCP}
    }\label{alg_fullNewton}
    \textbf{Input}: given the data $(M,q)$ of HLCP \eqref{eqn_HLCP}, let 
    $\beta=0.414213, \eta=\frac{\beta}{\sqrt{n+1}},\gamma=1-\eta$, and a stopping tolerance $\epsilon$, the required exact number of iterations is $\mathcal{N}=\left\lceil\frac{\log(\frac{n+1}{\epsilon})}{-\log(1-\frac{0.414213}{\sqrt{n+1}})}\right\rceil$
    \vspace*{.1cm}\hrule\vspace*{.1cm}
    \textbf{Initialize}: $\Bar{x}\leftarrow e, \Bar{s}\leftarrow e$, $\sigma\leftarrow\max(1,Mx+q, -x^\top M x - x^\top q)$, $M\leftarrow \frac{1}{\sigma}M$, $q\leftarrow\frac{1}{\sigma}q$;\\
    \textbf{for} $k=1,\cdots{}, \mathcal{N}$ \textbf{do}
    \begin{enumerate}[label*=\arabic*., ref=\theenumi{}]
        \item $\Bar{r}\leftarrow \Bar{s}-\psi(\Bar{x})$\\
        \item $\Bar{\mu}\leftarrow\frac{\Bar{x}^\top\Bar{s}}{n+1}$\\
        \item calculate $(d_{\Bar{x}}, d_{\Bar{s}})$ by solving (\ref{eqn_Newton_direction})\\
        \item $\Bar{x}\leftarrow\Bar{x}+d_{\Bar{x}}$\\
        \item $\Bar{s}\leftarrow\psi(\Bar{x})+\gamma\Bar{r}$\\
    \end{enumerate}
    \textbf{end}\\
    \textbf{return }$(\Bar{x},\Bar{s})$;
\end{algorithm}

\subsection{Convergence analysis}\label{sec_convergence_analysis}
Before showing the convergence analysis, we need to prove the following lemmas.

\begin{lemma}\label{lemma_equality}
     Let $\psi(\Bar{x})$ be given by \eqref{eqn_psi}. Then, 
        for any $\Bar{x}\in\mathbb{R}^{n+1}_{++}$,
    \begin{equation}
    \begin{aligned}
\Bar{x}^\top\psi(\Bar{x})&= 0,\\
\Bar{x}^\top\nabla \psi(\Bar{x}) &= - \psi(\Bar{x})^\top.
            \end{aligned}
        \end{equation}
\end{lemma}
\begin{proof}
    The proof is straightforward from the definitions $\psi(\Bar{x})$ \eqref{eqn_psi} and $\nabla \psi(\Bar{x})$ \eqref{eqn_nabla_psi}.
\end{proof}

\begin{lemma}\label{lemma_d_x_d_s}
    Let the direction $(d_{\Bar{x}},d_{\Bar{s}})$ be obtained from \eqref{eqn_Newton_direction}. Then, the equality
    \[
    d_{\Bar{x}}^\top d_{\Bar{s}} = d_{\Bar{x}}^\top \nabla\psi(\Bar{x}^k)d_{\Bar{x}} + \eta(1-\eta-\gamma)(\Bar{x}^k)^\top \Bar{s}^k
    \]
    holds. 
    
    Furthermore, by letting $\gamma=1-\eta$, we have
    \begin{equation}
\quad\left(\sqrt{\frac{\Bar{s}^k}{\Bar{x}^k}}\pdot d_{\Bar{x}}\right)^{\!\!\top} \left(\sqrt{\frac{\Bar{x}^k}{\Bar{s}^k}}\pdot d_{\Bar{s}}\right) = d_{\Bar{x}}^\top d_{\Bar{s}}\geq0    
    \end{equation}
    and
    \begin{equation}\label{eqn_core}
\left\|\sqrt{\frac{\Bar{s}^k}{\Bar{x}^k}}\pdot d_{\Bar{x}}\right\|^2 + \left\|\sqrt{\frac{\Bar{x}^k}{\Bar{s}^k}}\pdot d_{\Bar{s}}\right\|^2\leq\left\|\sqrt{\frac{\Bar{s}^k}{\Bar{x}^k}}\pdot d_{\Bar{x}}+ \sqrt{\frac{\Bar{x}^k}{\Bar{s}^k}}\pdot d_{\Bar{s}} \right\|^2.
    \end{equation}
\end{lemma}
\begin{proof}
    Premultiplying each side of (\ref{eqn_Newton_direction_1}) by $d_{\Bar{x}}^\top$ gives
\begin{equation}\label{eqn_dx_ds_eta}
    d_{\Bar{x}}^\top d_{\Bar{s}} - d_{\Bar{x}}^\top \nabla\psi(\Bar{x}^k)d_{\Bar{x}} = -\eta d_{\Bar{x}}^\top \Bar{r}^k
    \end{equation}
    and premultiplying each side of (\ref{eqn_Newton_direction_1}) by $(\Bar{x}^k)^\top$ gives
\begin{equation}\label{eqn_x_ds}
(\Bar{x}^k)^\top d_{\Bar{s}} -(\Bar{x}^k)^\top \nabla\psi(\Bar{x}^k)d_{\Bar{x}} = -\eta(\Bar{x}^k)^\top\Bar{r}^k.   
    \end{equation}
    By Lemma \ref{lemma_equality}, we have $(\Bar{x}^k)^\top \psi(\Bar{x}^k)=0$ and $(\Bar{x}^k)^\top \nabla\psi(\Bar{x}^k)=-\psi(\Bar{x}^k)^\top$,
    which can reduce \eqref{eqn_x_ds} to
\begin{equation}\label{eqn_x_ds_reduction}
    \begin{aligned}
     (\Bar{x}^k)^\top d_{\Bar{s}} +\psi(\Bar{x}^k)^\top d_{\Bar{x}} &= -\eta(\Bar{x}^k)^\top \Bar{r}^k\\
        &=-\eta(\Bar{x}^k)^\top\! \left(\Bar{s}^k-\psi(\Bar{x}^k)\right)\\
        &=-\eta(\Bar{x}^k)^\top \Bar{s}^k.
    \end{aligned} 
    \end{equation}
    Also, \eqref{eqn_dx_ds_eta} leads to
    \[
    \begin{aligned}
     d_{\Bar{x}}^\top d_{\Bar{s}} &= d_{\Bar{x}}^\top \nabla\psi(\Bar{x}^k)d_{\Bar{x}}-\eta d_{\Bar{x}}^\top\! \left(\Bar{s}^k-\psi(\Bar{x}^k)\right)  \\
        &=d_{\Bar{x}}^\top \nabla\psi(\Bar{x}^k)d_{\Bar{x}}-\eta\left(d_{\Bar{x}}^\top \Bar{s}^k-d_{\Bar{x}}^\top \psi(\Bar{x}^k)\right)\\
    &=d_{\Bar{x}}^\top \nabla\psi(\Bar{x}^k)d_{\Bar{x}}-\eta\left(d_{\Bar{x}}^\top \Bar{s}^k+(\Bar{x}^k)^\top d_{\Bar{s}}+\eta(\Bar{x}^k)^T\Bar{s}^k\right)\\
        &\quad\quad \text{(by \eqref{eqn_x_ds_reduction})}\\
    &=d_{\Bar{x}}^\top \nabla\psi(\Bar{x}^k)d_{\Bar{x}}-\eta\left[e^\top (\gamma\Bar{\mu}^ke-\Bar{x}^k\Bar{s}^k)+\eta(\Bar{x}^k)^\top \Bar{s}^k\right]\\
        &\quad\quad \text{(by \eqref{eqn_Newton_direction_2})}\\
    &=d_{\Bar{x}}^\top \nabla\psi(\Bar{x}^k)d_{\Bar{x}}-\eta\left[\gamma(\Bar{x}^k)^\top \Bar{s}^k-(\Bar{x}^k)^\top \Bar{s}^k + \eta(\Bar{x}^k)^\top \Bar{s}^k\right]\\
    &=d_{\Bar{x}}^\top \nabla\psi(\Bar{x}^k)d_{\Bar{x}}+\eta\left(1-\gamma-\eta\right)(\Bar{x}^k)^\top \Bar{s}^k,
    \end{aligned}
    \]
    which completes the first part of the proof. 
    
    Then, by letting $\gamma=1-\eta$ and by Lemma \ref{lemma_psd} ($\nabla \psi(\Bar{x}^k)\succeq0$), the above equality results in 
    \[
    \begin{aligned}
       &\quad\left(\sqrt{\frac{\Bar{s}^k}{\Bar{x}^k}}\pdot d_{\Bar{x}}\right)^{\!\!\top} \left(\sqrt{\frac{\Bar{x}^k}{\Bar{s}^k}}\pdot d_{\Bar{s}}\right) = d_{\Bar{x}}^\top d_{\Bar{s}} \\
&\qquad\qquad\qquad 
\qquad\qquad
\qquad\quad\ =d_{\Bar{x}}^\top \nabla\psi(\Bar{x}^k)d_{\Bar{x}}\geq0,
    \end{aligned}
    \]
    which can be used to derive the inequality
    \[
\left\|\sqrt{\frac{\Bar{s}^k}{\Bar{x}^k}}\pdot d_{\Bar{x}}\right\|^2 + \left\|\sqrt{\frac{\Bar{x}^k}{\Bar{s}^k}}\pdot d_{\Bar{s}}\right\|^2\leq\left\|\sqrt{\frac{\Bar{s}^k}{\Bar{x}^k}}\pdot d_{\Bar{x}}+ \sqrt{\frac{\Bar{x}^k}{\Bar{s}^k}}\pdot d_{\Bar{s}} \right\|^2.
    \]
This completes the second part of the proof.
\end{proof}

In our proposed Algorithm \ref{alg_fullNewton}, Step 4 ($\Bar{x}\leftarrow\Bar{x}+d_{\Bar{x}}$) and Step 5 ($\Bar{s}\leftarrow\psi(\Bar{x})+\gamma\Bar{r}$) are derived from the two different update equations for $\Bar{x}$ and $\Bar{s}$:
\begin{equation}\label{eqn_bar_x_update}
\Bar{x}^+\triangleq\Bar{x}^k+\alpha d_{\Bar{x}}>0
\end{equation}
and
\begin{equation}\label{eqn_bar_s_update}
    \begin{aligned}
    \Bar{s}^+&\triangleq\Bar{s}^k+\alpha d_{\Bar{s}} + \psi(\Bar{x}^+)-\psi(\Bar{x}^k)-\alpha\nabla\psi(\Bar{x}^k)d_{\Bar{x}}\\
&=\psi(\Bar{x}^+)+\left(\Bar{s}^k-\psi(\Bar{x}^k)\right) +\alpha\left(d_{\Bar{s}}-\nabla\psi(\Bar{x}^k)d_{\Bar{x}}\right)\\
&=\psi(\Bar{x}^+)+\Bar{r}^k-\alpha\eta\Bar{r}^k \\
        &= \psi(\Bar{x}^+)+ (1-\alpha\eta)\Bar{r}^k
    \end{aligned}
\end{equation}
when choosing the step size $\alpha=1$ and $\eta=1-\gamma$.

\begin{remark}\label{remark_residual}
    Thanks to this special update \eqref{eqn_bar_s_update} for $\Bar{s}$, which was also suggested in \cite{monteiro1990extension} for the feasible case, the next iterate of the residual vector $\Bar{r}^+$ satisfies
    \begin{equation}\label{eqn_residual_update}
    \begin{aligned}
        \Bar{r}^+ &\triangleq \Bar{s}^+ -\psi(\Bar{x}^+)\\
        &= (1-\alpha\eta)\Bar{r}^k.
    \end{aligned}
    \end{equation}
\end{remark}

Then, we have the following Lemma.
\begin{lemma}\label{lemma_x_s_update}
    If the next iterate $(\Bar{x}^+,\Bar{s}^+)$ is given by (\ref{eqn_bar_x_update}) and (\ref{eqn_bar_s_update}), then
    \[
    (\Bar{x}^+)^\top \Bar{s}^+ = \left(1-\alpha(1-\gamma)\right)(\Bar{x}^k)^\top \Bar{s}^k+\alpha^2\eta(1-\eta-\gamma)(\Bar{x}^k)^\top \Bar{s}^k.
    \]
    Furthermore, by letting $\eta=1-\gamma$, the residual $\Bar{r}$ (infeasibility measure) and the complementarity gap $(\Bar{x})^\top\Bar{s}$ are reduced at the same rate $(1-\alpha\eta)$.
\end{lemma}
\begin{proof}
Based on the updates of $\Bar{x}$ \eqref{eqn_bar_x_update} and $\Bar{s}$ \eqref{eqn_bar_s_update}, we have
    \[
    \begin{aligned}
        &\quad(\Bar{x}^+)^\top \Bar{s}^+ \\
        &=  (\Bar{x}^+)^\top \left(  \Bar{s}^k+\alpha d_{\Bar{s}} + \psi(\Bar{x}^+)-\psi(\Bar{x}^k)-\alpha\nabla\psi(\Bar{x}^k)d_{\Bar{x}}\right)\\
        &=(\Bar{x}^+)^\top \left(\Bar{s}^k+\alpha d_{\Bar{s}}\right) - (\Bar{x}^k+\alpha d_{\Bar{x}})^\top \left(\psi(\Bar{x}^k)+\alpha\nabla\psi(\Bar{x}^k)d_{\Bar{x}}\right)\\
        &=(\Bar{x}^+)^\top \left(\Bar{s}^k+\alpha d_{\Bar{s}}\right) - (\Bar{x}^k)^\top\psi(\Bar{x}^k) - \alpha(d_{\Bar{x}})^\top \psi(\Bar{x}^k) \\
        &\quad -\alpha(\Bar{x}^k)^\top \nabla\psi(\Bar{x}^k)d_{\Bar{x}} -\alpha^2d_{\Bar{x}}^\top \nabla\psi(\Bar{x}^k)d_{\Bar{x}}\\
         &\quad\quad (\text{by Lemma \ref{lemma_equality}})\\
        &=(\Bar{x}^+)^\top \left(\Bar{s}^k+\alpha d_{\Bar{s}}\right)  - \alpha^2d_{\Bar{x}}^\top \nabla\psi(\Bar{x}^k)d_{\Bar{x}}\\
        &=(\Bar{x}^k+\alpha d_{\Bar{x}})^\top \left(\Bar{s}^k+\alpha d_{\Bar{s}}\right)  - \alpha^2d_{\Bar{x}}^\top \nabla\psi(\Bar{x}^k)d_{\Bar{x}}\\
        &=(\Bar{x}^k)^\top \Bar{s}^k+\alpha\left(d_{\Bar{x}}^\top \Bar{s}^k+d_{\Bar{s}}^\top \Bar{x}^k\right) + \alpha^2\left( d_{\Bar{x}}^\top d_{\Bar{s}}-d_{\Bar{x}}^\top \nabla\psi(\Bar{x}^k)d_{\Bar{x}}\right)\\
        &=(\Bar{x}^k)^\top \Bar{s}^k+\alpha\left(d_{\Bar{x}}^\top \Bar{s}^k+d_{\Bar{s}}^\top \Bar{x}^k\right) + \alpha^2\eta(1-\eta-\gamma)(\Bar{x}^k)^\top \Bar{s}^k\\
        &\quad\quad (\text{by Lemma}~\ref{lemma_d_x_d_s})\\
        &=\left(1-\alpha(1-\gamma)\right)(\Bar{x}^k)^\top \Bar{s}^k + \alpha^2\eta(1-\eta-\gamma)(\Bar{x}^k)^\top \Bar{s}^k.\\
        &\quad\quad (\text{by \eqref{eqn_Newton_direction_2}})\\
    \end{aligned}
    \]
    which completed the first part of the proof. 
    
    Furthermore, by letting $\eta=1-\gamma$, the update of the complementarity gap $\Bar{x}^\top \Bar{s}$ is
\begin{equation}\label{eqn_x_s_update}
    (\Bar{x}^+)^\top \Bar{s}^+  = (1-\alpha\eta)(\Bar{x}^k)^\top \Bar{s}^k,
    \end{equation}
    which is the same as the update (\ref{eqn_residual_update}) of the infeasibility residual $\Bar{r}$. This completes the second part of the proof.
\end{proof}

\begin{lemma}\label{lemma_x_next_positive}
    Suppose that $\|\Bar{x}^k\pdot\Bar{s}^k-\Bar{\mu}^ke\|\leq\beta\Bar{\mu}^k$ where $0<\beta<\sqrt{2}-1$ (e.g., $\beta=0.414213$); letting $\eta=\frac{\beta}{\sqrt{n+1}}$, $\gamma=1-\eta$, and adopting the full Newton step ($\alpha=1$), then $\Bar{x}^+>0$.
\end{lemma}
\begin{proof}
Since $\|\Bar{x}^k\pdot\Bar{s}^k-\Bar{\mu}^ke\|\leq\beta\Bar{\mu}^k$ where $0<\beta<\sqrt{2}-1$, we have that
\[
\|\Bar{x}^k\pdot\Bar{s}^k-\Bar{\mu}^ke\|_\infty\leq\|\Bar{x}^k\pdot\Bar{s}^k-\Bar{\mu}^ke\|\leq\beta\Bar{\mu}^k,
\]
that is,
\[
-\beta\Bar{\mu}^ke\leq\Bar{x}^k\pdot\Bar{s}^k-\Bar{\mu}^ke\leq\beta\Bar{\mu}^ke
\]
Thus, we have
\begin{equation}\label{eqn_xs_geq_mu}
\min(\Bar{x}^k\pdot\Bar{s}^k)\geq(1-\beta)\Bar{\mu}^k.  
\end{equation}
As $\alpha=1$ and by Lemma \ref{lemma_x_s_update}, specifically \eqref{eqn_x_s_update}, we have 
\begin{equation}\label{eqn_mu_update}
  \Bar{\mu}^+=\gamma\Bar{\mu}^k,
\end{equation}
where $\Bar{\mu}^+=\frac{(\Bar{x}^+)^\top \Bar{s}^+}{n+1}$. Therefore, (\ref{eqn_Newton_direction_2}) is equivalent to the equality
\[
\Bar{s}^k\pdot d_{\Bar{x}} + \Bar{x}^k\pdot d_{\Bar{s}} = \Bar{\mu}^+ e -\Bar{x}^k\pdot\Bar{s}^k,
\]
which can also be reformulated as
\begin{equation}\label{eqn_sqrt_dx_ds}
\sqrt{\frac{\Bar{s}^k}{\Bar{x}^k}}\pdot d_{\Bar{x}}+ \sqrt{\frac{\Bar{x}^k}{\Bar{s}^k}}\pdot d_{\Bar{s}} = -\frac{\Bar{x}^k\pdot\Bar{s}^k-\Bar{\mu}^+ e}{\sqrt{\Bar{x}^k\pdot\Bar{s}^k}}.
\end{equation}
As a result,
\begin{equation}\label{eqn_d_x_d_s_inequality}
\begin{aligned}
&\left\|\sqrt{\frac{\Bar{s}^k}{\Bar{x}^k}}\pdot d_{\Bar{x}}+ \sqrt{\frac{\Bar{x}^k}{\Bar{s}^k}}\pdot d_{\Bar{s}} \right\|^2=\left\|\frac{\Bar{x}^k\pdot\Bar{s}^k-\Bar{\mu}^+ e}{\sqrt{\Bar{x}^k\pdot\Bar{s}^k}}\right\|^2\\
&\leq\frac{\|\Bar{x}^k\pdot\Bar{s}^k-\Bar{\mu}^+ e\|^2}{\min\!\big(\sqrt{\Bar{x}^k\pdot\Bar{s}^k}\,\big)^2}\leq\frac{\|\Bar{x}^k\pdot\Bar{s}^k-\Bar{\mu}^+ e\|^2}{(1-\beta)\Bar{\mu}^k}\quad\text{(by \eqref{eqn_xs_geq_mu})}\\
&=\frac{\|\Bar{x}^k\pdot\Bar{s}^k -\Bar{\mu}^ke+(1-\gamma)\Bar{\mu}^ke\|^2}{(1-\beta)\Bar{\mu}^k}\quad\text{(by \eqref{eqn_mu_update})}\\
&=\frac{\|\Bar{x}^k\pdot\Bar{s}^k -\Bar{\mu}^ke\|^2 + \|(1-\gamma)\Bar{\mu}^ke\|^2}{(1-\beta)\Bar{\mu}^k}\\
&\quad\text{(as $e^\top(\Bar{x}^k\Bar{s}^k -\Bar{\mu}^ke)=0$)}\\
&\leq\frac{\beta^2(\Bar{\mu}^k)^2 + \eta^2(n+1)(\Bar{\mu}^k)^2}{(1-\beta)\Bar{\mu}^k}\text{(as $\eta=\frac{\beta}{\sqrt{n+1}}$)}\\
&=\frac{2\beta^2(\Bar{\mu}^k)^2}{(1-\beta)\Bar{\mu}^k}=\frac{2\beta^2\Bar{\mu}^k}{1-\beta}.
\end{aligned}
\end{equation}

Therefore, by Lemma \ref{lemma_d_x_d_s} (namely \eqref{eqn_core}), we have that
\[
\left\|\sqrt{\frac{\Bar{s}^k}{\Bar{x}^k}}\pdot d_{\Bar{x}}\right\|^2 + \left\|\sqrt{\frac{\Bar{x}^k}{\Bar{s}^k}}\pdot d_{\Bar{s}}\right\|^2 \leq \frac{2\beta^2\Bar{\mu}^k }{1-\beta}.
\]
Then we have that
\[
\begin{aligned}
\left\| \frac{d_{\Bar{x}}}{\Bar{x}^k}\right\| &= \left\|\frac{1}{\sqrt{\Bar{x}^k\pdot\Bar{s}^k}}\sqrt{\frac{\Bar{s}^k}{\Bar{x}^k}}\pdot d_{\Bar{x}} \right\|\\
&\leq\frac{\left\|\sqrt{\frac{\Bar{s}^k}{\Bar{x}^k}}\pdot d_{\Bar{x}} \right\|}{\min\!\big(\sqrt{\Bar{x}^k\pdot\Bar{s}^k}\,\big)}  \leq
\frac{\sqrt{\left\|\sqrt{\frac{\Bar{s}^k}{\Bar{x}^k}}\pdot d_{\Bar{x}} \right\|^2+\left\|\sqrt{\frac{\Bar{x}^k}{\Bar{s}^k}}\pdot d_{\Bar{s}} \right\|^2}}{\min\!\big(\sqrt{\Bar{x}^k\pdot\Bar{s}^k}\,\big)}\\
&\leq\frac{\sqrt{\frac{2\beta^2\Bar{\mu}^k}{1-\beta}}}{\sqrt{(1-\beta)\Bar{\mu}^k}}\quad\quad\text{(by \eqref{eqn_xs_geq_mu}))}\\
&=\frac{\sqrt{2}\beta}{1-\beta}<1\quad\quad\text{(as $0<\beta<\sqrt{2}-1$)}
\end{aligned}
\]
Therefore, 
\[
\left\|\frac{d_{\Bar{x}}}{\Bar{x}^k}\right\|_\infty<1,
\]
which proves $\Bar{x}^+=\Bar{x}^k+d_{\Bar{x}}>0$. The proof is complete.
\end{proof}

\begin{lemma}\label{lemma_s_next_positive}
    Suppose that $\|\Bar{x}^k\pdot\Bar{s}^k-\Bar{\mu}^ke\|\leq\beta\Bar{\mu}^k$ where $0<\beta<\sqrt{2}-1$ (e.g., $\beta=0.414213$); letting $\eta=\frac{\beta}{\sqrt{n+1}}$, $\gamma=1-\eta$, and adopting the full Newton step ($\alpha=1$), then
    \[
    \|\Bar{x}^+\pdot\Bar{s}^+-\Bar{\mu}^+e\|\leq\beta\Bar{\mu}^+,
    \]
    where $\Bar{\mu}^+=\frac{(\Bar{x}^+)^\top \Bar{s}^+}{n+1}$,
    and $ \Bar{s}^+>0.$
\end{lemma}

\begin{proof}
By choosing $\alpha=1$ and according to the update of $\Bar{s}$ \eqref{eqn_bar_s_update}, we have
\[
\begin{aligned}
&\quad\Bar{x}^+\pdot\Bar{s}^+-\Bar{\mu}^+e \\
&= \Bar{x}^+\pdot\left(\Bar{s}^k+d_{\Bar{s}}+\psi(\Bar{x}^+)-\psi(\Bar{x}^k)-\nabla\psi(\Bar{x}^k)d_{\Bar{x}}\right) - \Bar{\mu}^+e\\
&=  (\Bar{x}^k+d_{\Bar{x}})\pdot(\Bar{s}^k+d_{\Bar{s}}) - \Bar{\mu}^+e \\
&\quad + \Bar{x}^+\pdot\left(\psi(\Bar{x}^+)-\psi(\Bar{x}^k)-\nabla\psi(\Bar{x}^k)d_{\Bar{x}}\right)\\
&= d_{\Bar{x}}\pdot d_{\Bar{s}} + \Bar{x}^+\pdot\left(\psi(\Bar{x}^+)-\psi(\Bar{x}^k)-\nabla\psi(\Bar{x}^k)d_{\Bar{x}}\right).\\
&\quad\quad\text{(by \eqref{eqn_Newton_direction_2} and \eqref{eqn_mu_update})}
\end{aligned}
\]
According to the definition of $\psi$ and $\nabla \psi$ (see \eqref{eqn_psi} and \eqref{eqn_nabla_psi}, respectively), and the full Newton update of $\Bar{x}^+=\Bar{x}^k+d_{\Bar{x}}$, we have that
\[
\begin{aligned}
    &\quad\psi(\Bar{x}^+)-\psi(\Bar{x}^k)-\nabla\psi(\Bar{x}^k)d_{\Bar{x}}\\
    &=\left[\begin{array}{cc}
         Mx^+ + q\tau^+ \\
         \frac{-(x^+)^\top Mx^+}{\tau^+} -q^\top x^+
    \end{array}\right]-\left[\begin{array}{cc}
         Mx^k + q\tau^k \\
         \frac{-(x^k)^\top Mx^k}{\tau^k} -q^\top x^k
    \end{array}\right]\\
   &\quad -
    \left[\begin{array}{cc}
             M & q \\
             -\frac{(x^k)^\top (M+M^\top)}{\tau^k}-q^\top  & \frac{(x^k)^\top Mx^k}{(\tau^k)^2}
         \end{array}\right]\left[\begin{array}{c}
              d_x  \\
              d_\tau 
         \end{array}\right]\\
         &=\left[\begin{array}{cc}
         Md_x + qd_\tau \\
         \frac{(x^k)^\top Mx^k}{\tau^k}-\frac{(x^+)^\top Mx^+}{\tau^+} -q^\top d_x
    \end{array}\right] \\
    &\quad- \left[\begin{array}{c}
         Md_x + qd_\tau    \\
         -\frac{(x^k)^\top (M+M^\top)d_x}{\tau^k}-q^\top d_x + \frac{d_\tau(x^k)^\top Mx^k}{(\tau^k)^2}
    \end{array}\right]\\
&=\left[\begin{array}{c}
         0  \\
    \frac{(x^k)^\top Mx^k+(x^k)^\top (M+M^\top)d_x}{\tau^k} -\frac{(x^+)^\top Mx^+}{\tau^+} - \frac{d_\tau(x^k)^\top Mx^k}{(\tau^k)^2}
\end{array}\right]\\
&\quad\quad\text{(as $d_x^\top(M+M^\top)x=2d_x^\top M x$)}\\
&=\left[\begin{array}{c}
         0  \\
\frac{d_\tau(x^k)^\top Mx^k+2d_\tau d_x^\top Mx^k-\tau^k d_x^\top Md_x}{\tau^k(\tau^k+d_\tau)} - \frac{d_\tau(x^k)^\top Mx^k}{(\tau^k)^2}
    \end{array}\right]
\end{aligned}
\]
and, by multiplying $\Bar{x}^+$, we have that
\[
\begin{aligned}
 &\quad \Bar{x}^+\pdot\left(\psi(\Bar{x}^+)-\psi(\Bar{x}^k)-\nabla\psi(\Bar{x}^k)d_{\Bar{x}}\right)\\   
 &=\left[\begin{array}{c}
     x^k+d_x  \\
    \tau^k+d_\tau 
 \end{array}\right] \\
&\quad \pdot \left[\begin{array}{c}
         0  \\
         \frac{d_\tau(x^k)^\top Mx^k+2d_\tau d_x^\top Mx^k-\tau^k d_x^\top Md_x}{\tau^k(\tau^k+d_\tau)} - \frac{d_\tau(x^k)^\top Mx^k}{(\tau^k)^2}
\end{array}\right]\\
&=\left[\begin{array}{c}
         0  \\
\frac{d_\tau(x^k)^\top Mx^k+2d_\tau d_x^\top Mx^k-\tau^k d_x^\top Md_x}{\tau^k}
\end{array}\right]\\
    &\quad-\left[\begin{array}{c}
         0  \\
\frac{d_\tau(\tau^k+d_\tau)(x^k)^\top Mx^k}{(\tau^k)^2}
\end{array}\right]\\
&=\left[\begin{array}{c}
         0  \\
\frac{\tau^kd_\tau(x^k)^\top Mx^k+2\tau^kd_\tau d_x^\top Mx^k-(\tau^k)^2 d_x^\top Md_x}{(\tau^k)^2}
\end{array}\right]\\
    &\quad-\left[\begin{array}{c}
         0  \\
\frac{(\tau^kd_\tau+d_\tau^2)(x^k)^\top Mx^k}{(\tau^k)^2}
\end{array}\right]\\
&=\left[\begin{array}{c}
         0  \\
\frac{2\tau^kd_\tau d_x^\top Mx^k-(\tau^k)^2 d_x^\top Md_x-d_\tau^2(x^k)^\top Mx^k}{(\tau^k)^2}
\end{array}\right]\\
&=\left[\begin{array}{c}
    0  \\
         -\left(d_x^\top Md_x-2\frac{d_\tau}{\tau^k}d_x^\top Mx^k+\frac{d_\tau^2}{(\tau^k)^2}(x^k)^\top Mx^k\right)
    \end{array}\right] \\
    &= \left[\begin{array}{c}
         0  \\
         -\left(d_x-\frac{d_\tau}{\tau^k}x^k\right)^\top M\left(d_x-\frac{d_\tau}{\tau^k}x^k\right)
    \end{array}\right].
\end{aligned}
\]
By the semi-positive definiteness of $M$ and Lemmas \ref{lemma_psd} and \ref{lemma_d_x_d_s}, we have that 
\[
\begin{aligned}
&\left\|\Bar{x}^+\pdot\left(\psi(\Bar{x}^+)-\psi(\Bar{x}^k)-\nabla\psi(\Bar{x}^k)d_{\Bar{x}}\right)\right\|\\
&= \left(d_x-\frac{d_\tau}{\tau^k}x^k\right)^{\!\!\top} M\left(d_x-\frac{d_\tau}{\tau^k}x^k\right)\\
&=d_{\Bar{x}}^\top \nabla\psi(\Bar{x}^k)d_{\Bar{x}}=d_{\Bar{x}}^\top d_{\Bar{s}}.
\end{aligned}    
\]

Then, by summarizing the above analysis, we have that
\[
\begin{aligned}
&\|\Bar{x}^+\pdot\Bar{s}^+-\Bar{\mu}^+e\|\\
&\leq\|d_{\Bar{x}}\pdot d_{\Bar{s}}|\| +\left\|\Bar{x}^+\pdot\left(\psi(\Bar{x}^+)-\psi(\Bar{x}^k)-\nabla\psi(\Bar{x}^k)d_{\Bar{x}}\right)\right\|\\
&\leq\|d_{\Bar{x}}\pdot d_{\Bar{s}}\| + d_{\Bar{x}}^Td_{\Bar{s}}=\left\|\left(\sqrt{\frac{\Bar{s}^k}{\Bar{x}^k}}\pdot d_{\Bar{x}}\right)\pdot\left(\sqrt{\frac{\Bar{x}^k}{\Bar{s}^k}}\pdot d_{\Bar{s}}\right)\right\| \\
&\quad+ \left(\sqrt{\frac{\Bar{s}^k}{\Bar{x}^k}}\pdot d_{\Bar{x}}\right)^{\!\!\top} \sqrt{\frac{\Bar{x}^k}{\Bar{s}^k}}\pdot d_{\Bar{s}}.
\end{aligned}
\]
For simplicity, we introduce two vectors $a\triangleq\sqrt{\frac{\Bar{s}^k}{\Bar{x}^k}}\pdot d_{\Bar{x}}$ and $b\triangleq\sqrt{\frac{\Bar{x}^k}{\Bar{s}^k}}\pdot d_{\Bar{s}}$; then we have that
\[
\begin{aligned}
&\quad\|\Bar{x}^+\pdot \Bar{s}^+-\Bar{\mu}^+e\|\leq\left\|a\pdot b\right\| + a^\top b\leq \left\|a\right\|\cdot \left\|b\right\| + a^\top b\\
&= \left\|a\right\|\cdot \left\|b\right\| + \tfrac{1}{2}\!\left\|a+b\right\|^2-\tfrac{1}{2}\big(\left\|a\right\|^2+\left\|b\right\|^2 \big)\\
&=\tfrac{1}{2}\!\left\|a+b\right\|^2 - \tfrac{1}{2}\!\left(\|a\|-\|b\| \right)^2\\
&\leq\tfrac{1}{2}\|a+b\|^2 =\tfrac{1}{2}\!\left\|\sqrt{\frac{\Bar{s}^k}{\Bar{x}^k}}\pdot d_{\Bar{x}}+\sqrt{\frac{\Bar{x}^k}{\Bar{s}^k}}\pdot d_{\Bar{s}} \right\|^2\\
&\leq\frac{\beta^2\Bar{\mu}^k}{1-\beta}=\frac{\beta}{(1-\beta)\gamma}(\beta\Bar{\mu}^+).\\
&\quad\quad\text{by \eqref{eqn_d_x_d_s_inequality} and \eqref{eqn_mu_update}}
\end{aligned}
\]
Therefore, to prove $\|\Bar{x}^+\Bar{s}^+-\Bar{\mu}^+e\|\leq\beta\Bar{\mu}^+$, we only need to prove the  inequality
\[
\frac{\beta}{(1-\beta)\gamma}\leq1,
\]
which is equivalent to proving that
\[
\frac{\beta}{1-\beta}\leq1-\frac{\beta}{\sqrt{n+1}}, \quad \forall n\geq1,
\]
as $\gamma=1-\eta=1-\frac{\beta}{\sqrt{n+1}}$. When $0<\beta<\sqrt{2}-1$, $\forall n\geq1$, the inequalities
\[
\frac{\beta}{1-\beta}+\frac{\beta}{\sqrt{n+1}}\leq\frac{\beta}{1-\beta}+\frac{\beta}{\sqrt{2}}<\frac{\sqrt{2}-1}{2-\sqrt{2}}+\frac{\sqrt{2}-1}{\sqrt{2}}=1
\]
hold, which completes the first part of the proof. 

Now we have that $\|\Bar{x}^+\pdot \Bar{s}^+-\Bar{\mu}^+e\|\leq\beta\Bar{\mu}^+$; thus
\[
\|\Bar{x}^+\pdot\Bar{s}^+-\Bar{\mu}^+e\|_\infty\leq\|\Bar{x}^+\pdot\Bar{s}^+-\Bar{\mu}^+e\|\leq\beta\Bar{\mu}^+.
\]
That is,
\[
-\beta\Bar{\mu}^+e\leq\Bar{x}^+\pdot\Bar{s}^+-\Bar{\mu}^+e\leq\beta\Bar{\mu}^+e,
\]
which clearly shows that $\min(\Bar{x}^+\pdot\Bar{s}^+)\geq(1-\beta)\Bar{\mu}^+>0$. 

By Lemma \ref{lemma_x_next_positive} we have $\Bar{x}^+>0$, so we have $\Bar{s}^+>0$, which completes the second part of the proof.
\end{proof}

Summarizing Lemmas \ref{lemma_x_next_positive} and \ref{lemma_s_next_positive}, if the current iterate  $(\Bar{x}^k,\Bar{s}^k)$ satisfies $\|\Bar{x}^k\Bar{s}^k-\Bar{\mu}^ke\|\leq\beta\Bar{\mu}^k$ (i.e., is located in a neighborhood of the central path), letting $\beta=0.414213$, $\eta=\frac{1}{\sqrt{n+1}}$, $\gamma=1-\eta$, $\alpha=1$, and adopting the update for $\Bar{x}$ \eqref{eqn_bar_x_update} and $\Bar{s}$ \eqref{eqn_bar_s_update}. The next iterate $(\Bar{x}^+,\Bar{s}^+)$ will keep positive and stay in the neighborhood of the central path. Clearly, the initialization strategy (\ref{eqn_initialization}) ensures that 
\begin{equation*}
\|\Bar{x}^0\pdot\Bar{s}^0-\Bar{\mu}^0\|=0\leq\beta\Bar{\mu}^0, \quad\Bar{x}^0>0, \quad\Bar{s}^0>0,
\end{equation*}
so the iterates, $(\Bar{x}^k,\Bar{s}^k)$ $\forall k=1,2,\cdots$ generated by Algorithm \ref{alg_fullNewton} satisfy
\[
    \begin{aligned}
        &\|\Bar{x}^k\pdot\Bar{s}^k-\Bar{\mu}^k\|\leq\beta\Bar{\mu}^k, \\
        &\Bar{x}^k>0,\quad \Bar{s}^k>0.
    \end{aligned}
    \]
Moreover, by Lemma \ref{lemma_x_s_update}, the infeasibility residual $\Bar{r}$ and the complementarity gap $\Bar{x}^\top \Bar{s}$ are both reduced at the same rate $(1-\eta)=1-\frac{0.414213}{\sqrt{n+1}}$. Therefore, by Lemma \ref{lemma_HLCP_central_path}, the iterates generated by Algorithm \ref{alg_fullNewton} will converge to a maximal complementary solution for the HLCP.

The number of iterations $\mathcal{N}$ required by Algorithm \ref{alg_fullNewton} will depend on the initial infeasibility residual $\Bar{r}^0$ and the complementarity gap $(\Bar{x}^0)^\top \Bar{s}^0$, which is discussed in the next Subsection.

\subsection{Scaling and initialization strategy}\label{sec_initialization_stragegy}
To develop a QP algorithm with \textit{data-independent} (\textit{only dimension-dependent}) iteration complexity, we first introduce a scaling strategy for the matrix $M$ and the vector $q$ in the LCP as 
\begin{equation}\label{eqn_scale}
    \Bar{M}\leftarrow \frac{1}{\sigma}M, \quad\Bar{q} \leftarrow\frac{1}{\sigma} q,
\end{equation}
where 
\[
    \sigma = \max(1,Me+q, -e^\top M e - e^\top q).
\]
This scaling strategy \eqref{eqn_scale} does not affect the optimal value of $x$ in (LCP). 

After that, we adopt the initialization strategy
\begin{equation}\label{eqn_initialization}
x^0=e,\quad\tau^0=1,\quad s^0=e, \quad \kappa^0=1.
\end{equation}
\begin{lemma}\label{lemma_initialization}
    The scaling strategy (\ref{eqn_scale}) and the initialization strategy (\ref{eqn_initialization}) ensure that the initial complementarity gap 
    \[
(\Bar{x}^0)^\top\Bar{s}^0=(x^0)^\top s^0+(\tau^0)^\top\kappa^0=n+1
    \]
    and the initial infeasibility residual
    \[
    \|\Bar{r}^0\|\leq n+1.
    \]
\end{lemma}
\begin{proof}
    Clearly, the initialization strategy (\ref{eqn_initialization}) results in the initial complementarity gap 
    \[
(\Bar{x}^0)^\top\Bar{s}^0=(x^0)^\top s^0+(\tau^0)^\top\kappa^0=n+1.
    \]
    
    After adopting the scaling strategy (\ref{eqn_scale}), the initial infeasibility residual
    \[
    \Bar{r}^0 = \left[\begin{array}{c}
         e - \Bar{M}e-\Bar{q}  \\
         1 + e^\top \Bar{M}e + e^\top \Bar{q}
    \end{array}\right]
    \]
    and the scaling factor $\sigma$ ensure that all elements of $\Bar{r}^0$ are nonnegative 
    (i.e.,  $\Bar{r}^0\geq0$). Thus, we have that
    \[
\|\Bar{r}^0\|^2\leq (e^\top \Bar{r}^0)^2 = (n+1)^2
    \]
    that is, $ \|\Bar{r}^0\|\leq n+1$. This completes the proof.
\end{proof}

Thus, our proposed scaling and initialization strategies make the initial complementarity gap $ (\Bar{x}^0)^\top\Bar{s}^0$ and initial infeasibility residual $\Bar{r}^0$ \textit{independent} from the problem data $(M,q)$ (\textit{only dependent on the problem dimension $n$}).

\subsection{Exact iteration complexity}\label{sec_exact_iteration_complexity}
Thanks to the above scaling and initialization strategies, our proposed infeasible full Newton IPM Algorithm \ref{alg_fullNewton} is \textit{only dimension dependent} (\textit{data independent}), is \textit{simple to calculate}, and has \textit{exact} iteration complexity.
\begin{theorem}\label{theorem_number_of_iterations}
The infeasible full Newton IPM Algorithm \ref{alg_fullNewton} returns a $\epsilon$-approximate maximal complementary solution of HLCP (\ref{eqn_HLCP}) in an exact 
\[
\mathcal{N}=\left\lceil\frac{\log(\frac{n+1}{\epsilon})}{-\log(1-\frac{0.414213}{\sqrt{n+1}})}\right\rceil
\]
iterations, where an $\epsilon$-approximate maximal complementary solution satisfies $\|\Bar{r}\|\leq\epsilon,~\Bar{x}^\top\Bar{s}\leq\epsilon$.
\end{theorem}
\begin{proof}
    Since our previous summarization of Lemmas \ref{lemma_x_next_positive} and \ref{lemma_s_next_positive} (see the last paragraph of the Subsection \ref{sec_convergence_analysis}), shows the convergence of Algorithm \ref{alg_fullNewton}, here we only show how to derive the required number of iterations for satisfying  $\|\Bar{r}^k\|\leq\epsilon$ and $(\Bar{x}^k)^\top\Bar{s}^k\leq\epsilon$.
    
     By Lemma \ref{lemma_x_s_update}, the infeasibility residual $\Bar{r}^k$ and complementarity gap $(\Bar{x}^k)^\top \Bar{s}^k$ at the $k$th iterate are \begin{equation}\label{eqn_residual_duality_gap_desreasing}
    \begin{aligned}
        \Bar{r}^k&=(1-\eta)^k\Bar{r}^0,\\
    (\Bar{x}^k)^\top \Bar{s}^k &= (1-\eta)^k(\Bar{x}^0)^\top \Bar{s}^0.
    \end{aligned}   
    \end{equation}
    Then, the inequalities $\|\Bar{r}^k\|\leq\epsilon$ and $(\Bar{x}^k)^\top\Bar{s}^k\leq\epsilon$ hold if
    \[
    \begin{aligned}
        (1-\eta)^k\|\Bar{r}^0\|&\leq\epsilon,\\
        (1-\eta)^k(\Bar{x}^0)^\top \Bar{s}^0&\leq\epsilon.
    \end{aligned}
    \]
    Taking logarithms gives
    \[
    \begin{aligned}
    &k\log(1-\eta) + \log(\|\Bar{r}^0\|)\leq\log(\epsilon),\\
    &k\log(1-\eta) + \log((\Bar{x}^0)^\top \Bar{s}^0)\leq\log(\epsilon).
    \end{aligned}
    \]
    Thus, the inequalities $\|\Bar{r}^k\|\leq\epsilon$ and $(\Bar{x}^k)^\top\Bar{s}^k\leq\epsilon$ hold if
    \[
    \begin{aligned}
k&\geq\left\lceil\frac{\log(\frac{\max(\|\Bar{r}^0\|,(\Bar{x}^0)^\top \Bar{s}^0)}{\epsilon})}{-\log(1-\eta)}\right\rceil=\left\lceil\frac{\log(\frac{n+1}{\epsilon})}{-\log(1-\eta)}\right\rceil.\\
&\quad\quad\text{(by Lemma \ref{lemma_initialization}, $\max(\|\Bar{r}^0\|,(\Bar{x}^0)^\top \Bar{s}^0)=n+1$)}
    \end{aligned}
    \]
The expressions \eqref{eqn_residual_duality_gap_desreasing} hold as equalities rather than inequalities, thus the number of iterations is \textit{exact}, rather than an upper bound. Therefore, after the exact number of iterations
\begin{equation}\label{eqn_number_of_iterations}
\mathcal{N}=\left\lceil\frac{\log(\frac{n+1}{\epsilon})}{-\log(1-\frac{0.414213}{\sqrt{n+1}})}\right\rceil,
\end{equation}
    the inequalities $\|\Bar{r}^k\|\leq\epsilon$ and $(\Bar{x}^k)^\top\Bar{s}^k\leq\epsilon$ hold, which completes the proof.
\end{proof}

\section{Algorithm implementation}\label{sec_algorithm_implementation}
\subsection{Robust procedure for solving Newton systems}
In  Algorithm \ref{alg_fullNewton}, the most computationally expensive step is Step 3, solving the linear system \eqref{eqn_Newton_direction}. To meet our mentioned execution time certificate requirement, we cannot adopt iterative methods to solve the linear system \eqref{eqn_Newton_direction} as their [flops] are hard to determine and \textit{data dependent}; here we have to use direct methods (such as LU, QR, and Cholesky decomposition) as their [flops] are determined and \textit{data independent}.

In fact, Algorithm \ref{alg_fullNewton} does not use $d_{\bar{s}}$ (as $\bar{s}\leftarrow\psi(\bar{x})+\gamma\bar{r}$, rather than $\bar{s}\leftarrow\bar{s}+d_{\bar{s}}$). By eliminating $d_{\Bar{s}}$ with
\[
d_{\Bar{s}} =\frac{\Bar{s}^k}{\Bar{x}^k}d_{\Bar{x}^k} + \gamma\Bar{\mu}^k\frac{1}{\Bar{x}^k}-\Bar{s}^k,
\] we can reduce the linear system \eqref{eqn_Newton_direction} to the compact system
\begin{equation}\label{eqn_campcat_Newton_system}
    \left(\nabla \psi(\Bar{x}^k) + \diag\!\left(\frac{\Bar{s}^k}{\Bar{x}^k}\right)\right)d_{\Bar{x}} = \gamma\Bar{\mu}^k\frac{1}{\Bar{x}^k} - \Bar{s}^k + \eta\Bar{r}^k.
\end{equation}
\begin{remark}\label{remark_psd}
    By Lemma \ref{lemma_psd}, $\nabla \psi(\Bar{x}^k)$ is semi-positive definite. Thus, for arbitrary $\bar{x}^k,\bar{s}^k\in\mathbb{R}^{n+1}_{++}$, \eqref{eqn_campcat_Newton_system} is a positive definite system and thus there exists a unique solution.
\end{remark}

After adopting the scaling \eqref{eqn_scale} and according to the definition $\nabla\psi$ \eqref{eqn_nabla_psi}, $M$, and $\bar{x}=\mathrm{col}(z,y,\tau),\bar{s}=\mathrm{col}(v,w,\kappa)$, the linear system \eqref{eqn_campcat_Newton_system} is expanded as
\begin{equation}\label{eqn_expanded_Newton}
    \begin{aligned}
&\left[\begin{array}{ccc}
Q(t)+\diag(\frac{v^k}{z^k}) & -A(t)^\top & c(t)\\
            A(t) & \diag(\frac{w^k}{y^k}) & -b(t)\\
            -\frac{2(z^k)^\top Q(t)}{\tau^k}-c(t)^\top & b(t)^\top & \frac{(z^k)^\top Q(t)z^k}{(\tau^k)^2}+\frac{\kappa^k}{\tau^k} 
\end{array}\right]\\
        &\quad\cdot \left[\begin{array}{c}
                 d_z  \\
                 d_y\\
                 d_\tau
\end{array}\right]=\left[\begin{array}{c}
\gamma\bar{\mu}\frac{1}{z^k}-v^k +\eta r_z  \\
\gamma\bar{\mu}\frac{1}{y^k}-w^k +\eta r_y\\
\gamma\bar{\mu}\frac{1}{\tau^k}-\kappa^k +\eta r_\tau 
\end{array}\right]
    \end{aligned}
\end{equation}
where $r_z\triangleq v^k-Q(t)z^k+A(t)^\top y^k-c(t)\tau^k$, $r_y\triangleq w^k-A(t)z^k+b(t)\tau^k$, $r_\tau\triangleq \kappa^k+(z^k)^\top Q(t) z^k+c(t)^\top z^k-b(t)^\top y^k$, and $\bar{\mu}\triangleq \frac{(z^k)^\top v^k+y^k)^\top w^k+\tau^k\kappa^k}{n+1}$. 

The linear system \eqref{eqn_expanded_Newton} is not a symmetric positive definite system, and we tried to reduce it to a smaller symmetric positive definite system to allow the use of an efficient Cholesky decomposition method, but this approach is not robust and cannot provide a high-accuracy solution when close to the end of the iterations. This occurs because, when close to the end of the iterations, the linear system \eqref{eqn_expanded_Newton} is too ill-conditioned, which is the cost of the use of the homogeneous formulation. For solving \eqref{eqn_expanded_Newton}, we choose to call Lapack's function \textit{dgesv} \cite{laug}, which uses the LU decomposition with partial pivoting and row interchanges. 

Future work will investigate an effective preconditioner to allow the use of Cholesky decomposition (which has half of the cost of LU factorization) on the smaller compact system. Currently, we observe that the diagonal preconditioner reduces the condition number of \eqref{eqn_expanded_Newton}.

\subsection{Infeasibility reporting criteria}
 As for the iterates $\{x^k,\tau^k,s^k,\kappa^k\}$ generated by Algorithm \ref{alg_fullNewton}, Lemma \ref{lemma_HLCP_solution} indicates that $\tau^k$ will be bounded away from zero and $\kappa^k$ will quadratically converge to zero if the QP (\ref{eqn_QP}) is feasible, or $\kappa^k$ will be bounded away from zero and $\tau^k$ will quadratically converge to zero if the QP (\ref{eqn_QP}) is infeasible. Thus, this paper adopts $\tau<\kappa$ as the infeasibility reporting criteria.

\subsection{Execution time certificate}\label{sec_alg_2}
Summarizing the above analysis, we rewrite Algorithm \ref{alg_fullNewton} as an implementable Algorithm \ref{alg_efficient}. Note that Step 1 of Algorithm \ref{alg_fullNewton} is equal to Step 5 (before the iteration) and Step 6.6 of Algorithm \ref{alg_efficient}. In addition to using the Lapack library, we also use the efficient BLAS library \cite{lawson1979basic} to implement matrix-vector multiplications and vector inner product operations of  Algorithm \ref{alg_efficient}. 

Thanks to the \textit{exact} iteration complexity and the use of LU decomposition in solving the linear system \eqref{eqn_expanded_Newton}, the total [flops] of Algorithm \ref{alg_efficient} is easily countable and fixed. In practice, running Algorithm \ref{alg_efficient} once on a given processor, the execution time is obtained and will almost surely be the same among all QPs with the same problem dimension.

\floatname{algorithm}{Algorithm}
\begin{algorithm}
    \caption{EIQP: execution-time-certified and infeasibility-detecting QP solver,
    an infeasible homogeneous full Newton IPM algorithm for convex QP \eqref{eqn_QP}.
    }\label{alg_efficient}
    \textbf{Input}: given the data of QP (\ref{eqn_QP}) $(Q(t),c(t),A(t),b(t))$, $\beta=0.414213, \eta=\frac{\beta}{\sqrt{n+1}},\gamma=1-\eta$, and a optimality level $\epsilon$, and then the required exact number of iterations is $\mathcal{N}=\left\lceil\frac{\log(\frac{n+1}{\epsilon})}{-\log(1-\frac{0.414213}{\sqrt{n+1}})}\right\rceil$
    \vspace*{.1cm}\hrule\vspace*{.1cm}
    \begin{enumerate}[label*=\arabic*., ref=\theenumi{}]
        \item $z\leftarrow e,~y \leftarrow e,~\tau\leftarrow 1,~v\leftarrow e,~w\leftarrow e,~\kappa\leftarrow 1$;\\
        \item cache: $Qz\leftarrow Q(t)z, Ay\leftarrow A(t)^\top y, Az\leftarrow A(t)z, zc\leftarrow z^\top c(t), yb\leftarrow y^\top b(t),zQz\leftarrow z^\top Qz$;\\
        \item $\sigma\leftarrow\max(1,Qz-Ay+c(t),Az-b(t),-zQz-zc+yb)$;\\
        \item scale: $Q(t)\leftarrow \frac{1}{\sigma}Q(t), c(t)\leftarrow\frac{1}{\sigma}c(t),A(t)\leftarrow\frac{1}{\sigma}A(t),b(t)\leftarrow\frac{1}{\sigma}b(t),Qz\leftarrow\frac{1}{\sigma}Qz,Ay\leftarrow\frac{1}{\sigma}Ay,Az\leftarrow\frac{1}{\sigma}Az,zc\leftarrow\frac{1}{\sigma}zc,yb\leftarrow\frac{1}{\sigma}yb,zQz\leftarrow\frac{1}{\sigma}zQz$;\\
        \item $r_z\leftarrow v-Qz+Ay-c(t)\tau$, $r_y\leftarrow w-Az+b(t)\tau$, $r_\tau\leftarrow \kappa + \frac{zQz}{\tau} + zc - yb$;\\
        \item \textbf{for} $k=1,\cdots{}, \mathcal{N}$ \textbf{do}
        \begin{enumerate}[label*=\arabic*., ref=\theenumi{}]
            \item $\Bar{\mu}\leftarrow\frac{z^\top v + y^\top w +\tau\kappa}{n+1}$;\\
            \item call \textit{dgesv} to solve \eqref{eqn_expanded_Newton}, obtain $d_z, d_y, d_\tau$;\\
            \item $z\leftarrow z+d_{z},~y\leftarrow y+d_y,~\tau\leftarrow\tau+d_\tau$;\\
            \item cache: $Qz\leftarrow Q(t)z, Ay\leftarrow A(t)^\top y, Az\leftarrow A(t)z, zc\leftarrow z^\top c(t), yb\leftarrow y^\top b(t),zQz\leftarrow z^\top Qz$;\\
            \item $v\leftarrow Qz-Ay+c(t)\tau+\gamma r_z$, $w\leftarrow Az-b(t)\tau+\gamma r_y, ~\kappa\leftarrow -\frac{zQz}{\tau}-zc+yb+\gamma \Tilde{r}$;\\
            \item $r_z\leftarrow\gamma r_z, ~r_y\leftarrow\gamma r_y, r_\tau\leftarrow \gamma r_\tau$;
        \end{enumerate}
        \item[~] \textbf{end}\\
        \item \textbf{if} $\tau<\kappa$, \textbf{return}~ (\ref{eqn_QP}) is infeasible;~\textbf{otherwise}, \\
        \item \textbf{return } $z= \frac{1}{\tau}z$.
    \end{enumerate}
\end{algorithm}

\section{Numerical Experiments}
This section first compares the infeasibility-detection capabilities of Algorithm \ref{alg_efficient} with the known MATLAB's \textbf{quadprog} solver and popular \textbf{OSQP} solver. Then Algorithm \ref{alg_efficient} is applied to the AFTI-16 MPC example and ACC CBF-QP example.  The reported simulation results were obtained on a Mac mini with an Apple M4 Chip (10-core CPU and 16 GB RAM). 

The C-code implementation of Algorithm \ref{alg_efficient} is just one single C file, making it easy to integrate with other software and appealing in deploying in embedded production systems. Its MATLAB, Julia, and Python interfaces and numerical examples are publicly available at \url{https://github.com/liangwu2019/EIQP}.

\subsection{Random infeasible QP examples}
To demonstrate the infeasibility-detection capability of our proposed Algorithm \ref{alg_efficient}, we randomly generate 100 QPs ($\min \frac{1}{2}z^\top Qz+c^\top z, \text{ s.t. } Az\leq b$) for each case: condition number $cond(Q)$ varies from $10^1$ to $10^6$. To make QPs infeasible, we then add the contradictory constraints: $A\leftarrow [A;-A(1,:);-A(2,:)],b\leftarrow [b;-b(1)-1;-b(2)-1]$. As for solver settings, ``quadprog" uses its default setting, ``OSQP-default" uses its default setting, ``OSQP-$5e5$" changes the maximum number of iterations to $5e5$, and other settings as ('eps\_abs',1e-10,'eps\_rel',1e-10,'polish',true), and our ``EIQP" chooses $\epsilon=1e$-6. The infeasibility detection rate of these QP solvers is plotted in Fig.\  \ref{Fig1}. 

The quadprog solver has poor infeasibility-detection results and its detection rate gradually becomes lower as the condition number of $Q$ increases. The OSQP solver with the default setting cannot detect 100\% of QP infeasibility; increasing the maximum number of iterations to $5e5$ and decreasing the optimality level can increase its detecting rate but still cannot reach 100\%. It is hard to tune the OSQP algorithm parameters to achieve 100\% detection rate which is also computationally expensive. Our EIQP solver has a 100\% infeasibility detection rate among all cases, consistent with the proof.
\begin{figure}[!ht]
        \hspace*{-0em}\includegraphics[width=1.0\columnwidth]{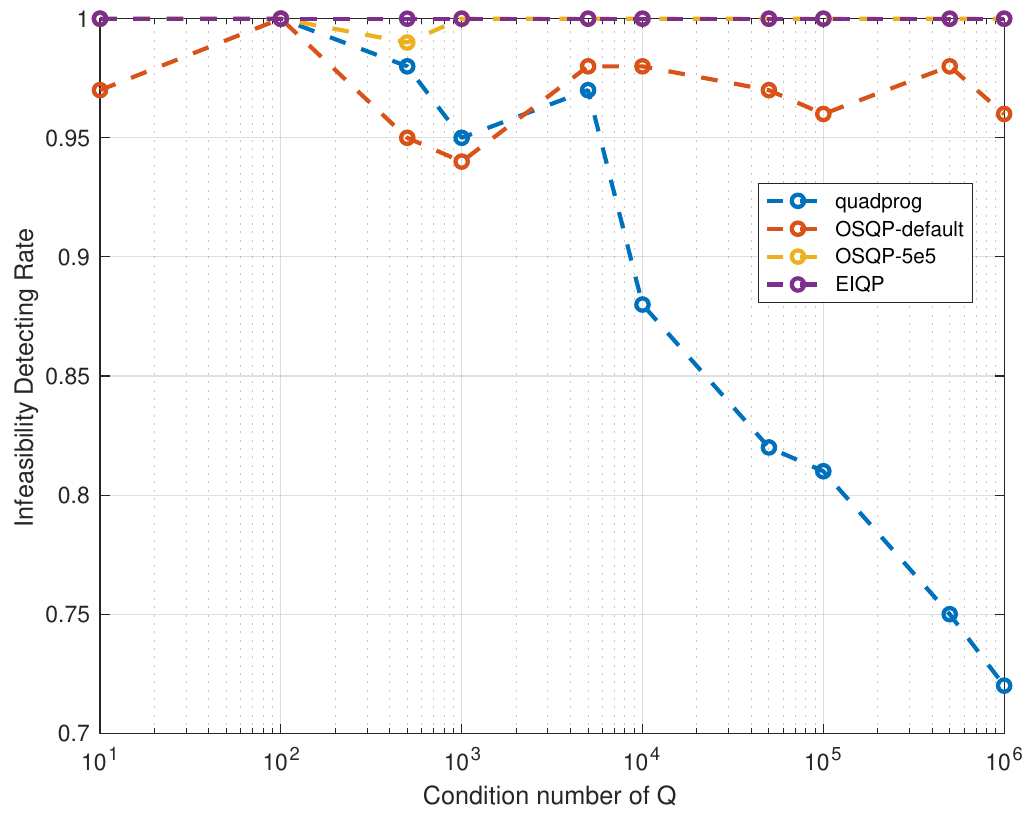} 
        
        \vspace{-0.2cm}
        \caption{Comparison of quadprog, OSQP-default, OSQP-5e5, and EIQP solvers in detecting
        infeasibility of random infeasible QP problems as a function of condition number of the Hessian matrix $Q$.}
        \label{Fig1} 
\end{figure}

\subsection{Real-time MPC applications}\label{subsection_MPC}
A general linear MPC for tracking is
\begin{equation}\label{problem_MPC_tracking}
    \begin{aligned}
        \min~&\frac{1}{2}\sum_{k=0}^{N_p-1}\left\|Cx_{k+1}-r(t))\right\|_{W_y}^{2} + \left\|u_{k} - u_{k-1}\right\|_{W_{\Delta u}}^{2} \\
        \mathrm{s.t.}~&x_{k+1}=A x_{k}+B u_{k},~k=0, \ldots, N_p-1,\\
        & E_x x_{k+1} + E_u u_k\leq f_k,~k=0, \ldots, N_p-1,\\
        &x_0 = x(t), \quad u_{-1} = u(t-1),
    \end{aligned}
\end{equation}
where $x_k\in\mathbb{R}^{n_x}$ denotes the states and $u_k\in\mathbb{R}^{n_u}$ denotes the control inputs. The control objective is to minimize the tracking errors between the output $Cx_{k+1}$ and the reference signal $r(t)$, while penalizing the control input increments $\Delta u_k$, along the prediction horizon $N_p$. The matrices $W_y\succ0$ and $W_{\Delta u}\succ0$ are the weights for the output tracking error and control input increments, respectively. $x(t)$ and $u(t-1)$ denote the current feedback states and previous control inputs, respectively. $E_x x_{k+1} + E_u u_k \leq f_k$ denotes the general constraints formulation such as the box constraints $u_{\min } \leq u_{k} \leq u_{\max}, \ x_{\min} \leq x_{k+1} \leq x_{\max}$ or the terminal constraints $x_{N_p}\in\mathcal{X}_{N_p}$. Denoting
\[
z\triangleq\mathrm{col}(u_0-u_{\min}, \cdots{}, u_{N_p-1}-u_{\min})\in\mathbb{R}^{n_z}
\]
($n_z=N_pn_u$), we can eliminate the states (see \cite{jerez2011condensed}), to get a compact QP \eqref{eqn_QP}. Here we consider the classical AFTI-16 example from \cite{KAS88,bemporad1997nonlinear}:
\[
\left\{\begin{aligned}
\dot{x} =&{\left[\begin{array}{cccc}
-0.0151 & -60.5651 & 0 & -32.174 \\
-0.0001 & -1.3411 & 0.9929 & 0 \\
0.00018 & 43.2541 & -0.86939 & 0 \\
0 & 0 & 1 & 0
\end{array}\right]} x\\&+{\left[\begin{array}{cc}
-2.516 & -13.136 \\
-0.1689 & -0.2514 \\
-17.251 & -1.5766 \\
0 & 0
\end{array}\right] }u \\
y =&{\left[\begin{array}{llll}
0 & 1 & 0 & 0 \\
0 & 0 & 0 & 1
\end{array}\right]x}
\end{aligned}\right.
\]
which is sampled using zero-order hold every 50~ms. Our MPC controller feedback time is also 50~ms, so we need to ensure that the MPC design meets this execution time requirement. The input constraints are $|u_i| \leq 25^{\circ},i = 1, 2$, the output constraints are $-0.5\leq y_1 \leq 0.5 $ and $-100 \leq y_2 \leq 100$. The control goal is to make the pitch angle $y_2$ track a reference signal $r_2$. In designing the MPC controller, we take $W_y = \diag$([10,10]), $W_u = 0$, and $W_{\Delta u}= \diag$([0.1, 0.1]). 

Now, the remaining MPC setting is to choose the prediction horizon to meet the requirement that the execution time is smaller than the feedback time 50~ms and maintain good closed-loop performance. 

We test our EIQP solver ($\epsilon=10^{-8}$) for different prediction horizon settings: $N_p=5,10,15,20,25$, and their execution times, MPC closed-loop cost (average), and MPC constraint violations are listed in Table \ref{tab1}. For the prediction horizon $N_p \leq 20$, the execution time of our EIQP is less than the feedback time of $50$ ms. As the prediction horizon $N_p$ increases, the MPC tracking costs decrease but the MPC gradually violates the constraints in $N_p\geq15$ settings, this is because as the prediction horizon $N_p$ increases, the resulted QPs become gradually ill-conditioning, and thus requires smaller optimality level, which can reduce the constraint violations to 0 but increase the execution time. In our given processor, $N_p=10$ should be chosen. The closed-loop simulation results for $N_p=10$ are plotted in Fig.\ \ref{Fig2}, which shows that the pitch angle correctly tracks the reference signal from $0^{\circ}$ to $10^{\circ}$ and then back to $0^{\circ}$, and that both the input and output constraints are satisfied.
\begin{table}[!htbp]
\caption{EIQP ($\epsilon=10^{-8}$) computational performance for MPC with different prediction horizon $(N_p=5,10,15,20,25)$ settings}

\vspace{-0.2cm}

\centering
\begin{tabular}{ccccc}
\toprule
Prediction  & Execution  & Less than  & MPC & MPC constraint\\
horizon &  time [ms] & 50 [ms]? & costs (avg) &  violations (avg)
\\\midrule
$N_p=5$ & 1.6 & \tikzcmark  & 42.6218 & 0.0\\
$N_p=10$ & 8.7 & \tikzcmark  & 42.5558 & 0.0\\
$N_p=15$ & 22.8 & \tikzcmark & 42.5336 & \textbf{0.0001}\\
$N_p=20$ & 47.6 & \tikzcmark & 42.4517 & \textbf{0.0009}\\
$N_p=25$ & \textbf{86.3} & \tikzxmark & 41.4774 & \textbf{0.0103}\\
\bottomrule
\end{tabular}
\label{tab1}
\end{table}

\begin{figure}[!ht]
        \hspace*{-0em}\includegraphics[width=1\columnwidth]{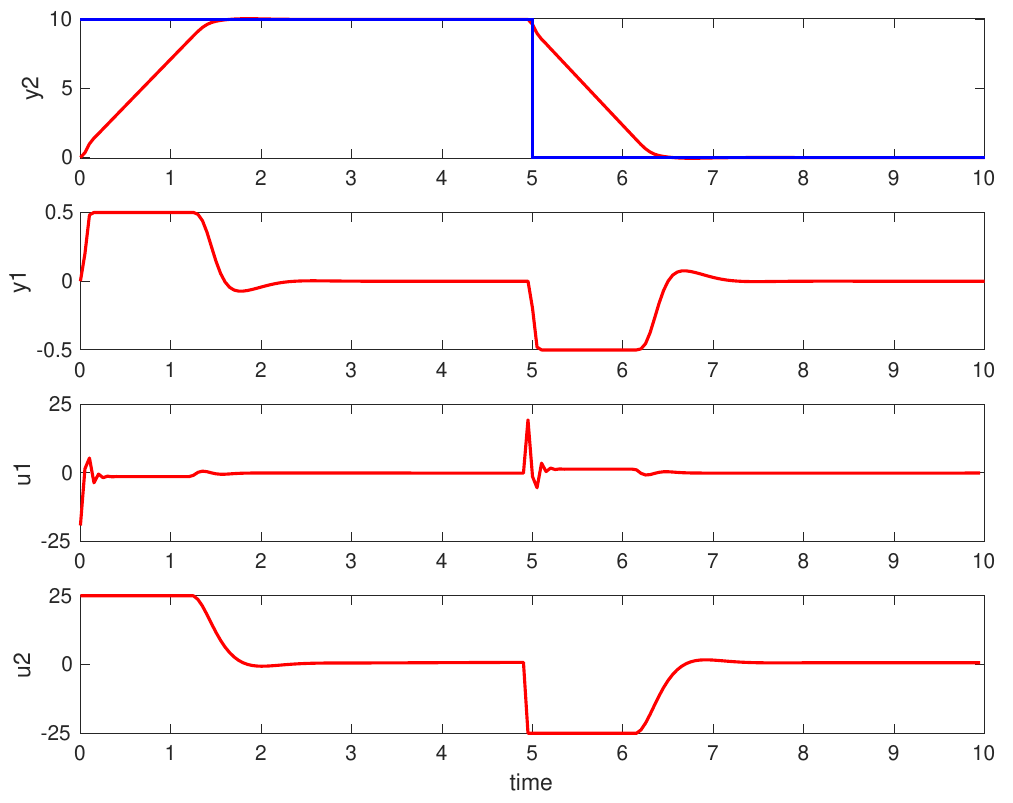} 
        
        \vspace{-0.3cm}
        \caption{Closed-loop performance of AFTI-16 under the MPC prediction horizon $N_p=10$ setting and EIQP ($\epsilon=10^{-8}$).}
        \label{Fig2} 
\end{figure}

\subsection{Real-time CLF-HOCBF-QP applications}\label{subsection_CBF_QP}

This subsection uses of EIQP to solve High-Order Control Barrier Function (HOCBF)-based QP \cite{xiao2021high}, and compares it to the interior point and active set algorithms in QuadProg.

HOCBF has been shown to be capable of transforming nonlinear optimizations for nonlinear systems and constraints into QPs with formal guarantees \cite{xiao2023safe}. We consider the Adaptive Cruise Control (ACC) problem with nonlinear vehicle dynamics \cite{Khalil2002}
\begin{equation}\label{eqn:vehicle}
\underbrace{\left[\begin{array}{c} 
	\dot z\\
	\dot v
	\end{array} \right]}_{ \dot {\bm x}}=
\underbrace{\left[\begin{array}{c}  
	v_p - v\\
	-\frac{1}{M}F_{r}(v)
	\end{array} \right]}_{f(\bm x)} + 
\underbrace{\left[\begin{array}{c}  
	0\\
	\frac{1}{M}
	\end{array} \right]}_{g(\bm x)}u
\end{equation}
where $z$ denotes the distance between the ego vehicle and its preceding vehicle; $v, v_p$ denote the speed of the ego and its preceding vehicle, respectively; $u$ is the control (acceleration) of the ego vehicle; $M$ is the mass of the ego vehicle, and $F_r(v) = f_0\mathrm{sgn}(v) + f_1v + f_2 v^2$, where $f_0 > 0$, $f_1 > 0$, and $f_2 > 0$ are 
scalars. The first term in $F_{r}(v(t))$ denotes the coulomb friction force, the second term denotes the viscous friction force, and the last term denotes the aerodynamic drag.

$\mathbf{Constraint\
 1}$ (Control bounds): There are constraints
on control for the ego vehicle, i.e.,
\begin{equation}\label{eqn:limitation}%
\begin{aligned}
-c_dMg\leq u\leq c_aMg,  \end{aligned} 
\end{equation}
where $c_d>0$ and $c_a > 0$ are deceleration and
acceleration coefficients, respectively, and $g$ is the gravity constant.

$\mathbf{Constraint\ 2}$ (Safety constraint): The distance $z$ is required to satisfy
\begin{equation}\label{eqn:safetyACC}%
z\geq \delta, 
\end{equation}
where $\delta>0$ is determined by the length of the two vehicles.

$\mathbf{Objective\ 1}$ (Desired speed): The ego vehicle always attempts to achieve a desired speed $v_d > 0$.

$\mathbf{Objective\ 2}$ (Minimum control effort): The control input effort of the ego vehicle,
\begin{equation}\label{eqn:energy}
J_i(u_i(t))=\int_0^T\! \left(\frac{u(t) - F_{r}(v(t))}{M}\right)^{\!\!2}dt,
\end{equation}
should be minimized, 
where $T > 0$ is the final time.

\textbf{Problem:}
The ACC problem is to determine control laws to achieve Objectives 1 and 2 subject to Constraints 1 and 2 for the ego vehicle governed by dynamics (\ref{eqn:vehicle}).

We use the Control Lyapunov Function (CLF) \cite{ames2014rapidly} to achieve the desired speed and use the HOCBF to enforce the safety constraint (\ref{eqn:safetyACC}). Then the ACC problem can be transformed into the QP that is solved at every time step:
\begin{equation} \label{eqn:objACC}
\bm u^*(t) = \arg\min_{\bm u(t)} \frac{1}{2}\bm u(t)^\top H\bm u(t) + F^\top \bm u(t)
\end{equation}
\[
\bm u(t)\! =\! \left[\begin{array}{c}
\! u(t)\!\\
\!\delta_{acc}(t)\!
\end{array} \right]\!,
H\! =\! \left[\begin{array}{cc} 
\frac{2}{M^2} & 0\\
0 & 2p_{acc}
\end{array} \right]\!, F\! =\!  \left[\begin{array}{c} 
\!\frac{-2F_r(v(t))}{M^2}\!\\
0
\end{array} \right]. 
\]
subject to
$$
\begin{aligned}
L_f^2b(\bm x) + LgL_f b(\bm x) + 2L_fb(\bm x) + b(\bm x) \geq 0,\\
L_fV(\bm x) + L_gV(\bm x) + \epsilon_{acc} V(\bm x) \leq \delta_{acc},\\
-c_dMg\leq u\leq c_aMg,  
\end{aligned}
$$
where $b(\bm x) = z - \delta$ is the HOCBF, $V(\bm x) = (v - v_d)^2$ is the CLF; $L_f, L_g$ denotes the Lie derivative along $f, g$, respectively; $\epsilon_{acc} > 0$; and $\delta_{acc}$ is a relaxation variable that ensures the CLF will not conflict with the HOCBF constraint.

 The simulation parameters are listed in Table \ref{table:param}. The initial conditions are $z(0) = 100$\,m, $v(0) = 20$\,m/s, and $v_p(t) = 13.89\, \mathrm{m/s},\  \forall t\in[0, T]$.

\begin{table}
\caption{Simulation parameters for ACC}\label{table:param}
    \vspace{-0.4cm}
    
	\begin{center}
	\begin{tabular}{|c||c||c|}
			\hline
Parameter & Value & Units\\
			\hline
			\hline
			$\delta$ & 10& m\\
            \hline
			$v_d$ & 24& m/s\\
			\hline
			$M$ & 1650& kg\\
			\hline
			g & 9.81& m/s$^2$\\
			\hline
			$f_0$ & 0.1& N\\
			\hline
			$f_1$ & 5& Ns/m\\
			\hline
			$f_2$ & 0.25& Ns$^2$/m\\
			\hline
			$\epsilon_{acc}$ & 10& unitless\\
			\hline
			$c_a, c_d$ & 0.4& unitless\\
			\hline
			$p_{acc}$ & 1& unitless\\
			\hline
		\end{tabular}
	\end{center}
	
\end{table}

\textbf{Simulation results.} The simulation results are given in Table \ref{tab2} and Fig.\  \ref{Fig3}. The proposed EIQP is around 5 times faster than the QuadProg in Matlab, while being able to detect the infeasible case when $c_d$ decreases to 0.375. All the solvers can ensure the correctness of the solutions and guarantee the safety of the ACC problem, as shown by $b(\bm x(t))\geq 0$ for all $t\in[0,T]$ in Fig.\  \ref{Fig3}. In our processor, EIQP certifies that its execution time is below 0.075 [ms], enabling our CLF-HOCBF-QP controller to operate at a maximum frequency of 13.3 kHz. Moreover, the execution time of EIQP solver has the smallest standard deviation, which is consistent with our theoretical result: EIQP performs \textit{exact} and \textit{fixed} iterations.

\begin{table}[!htbp]
\caption{EIQP comparison with QuadProg (interior point, active set) in CLF-HOCBF-QP}
\centering
\begin{tabular}{ccccc}
\toprule
Method  & Execution time & HOCBF  & Infeasibility\\
 &  [ms] & $b(\bm x(T))$ & detected?
\\\midrule
Active set & 0.297$\pm 0.252$ & 2.010e-7  & \tikzxmark\\
Interior point & 0.361$\pm 0.196$ & 1.825e-7  & \tikzcmark\\
EIQP & \textbf{0.065}$\pm$ \textbf{0.008} & 1.812e-7 & \tikzcmark\\
\bottomrule
\end{tabular}
\label{tab2}
\end{table}

\begin{figure}[!ht]
        \hspace*{-0em}\includegraphics[width=1\columnwidth]{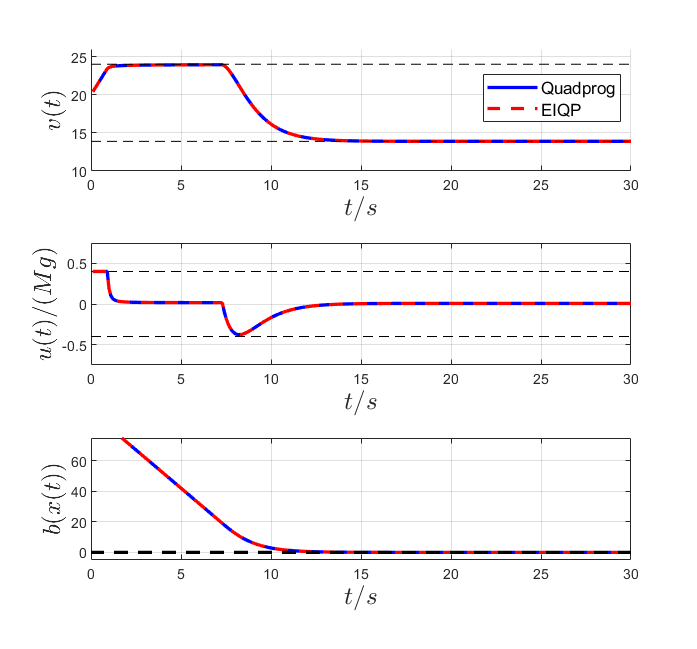} 

        \vspace{-0.7cm}
        
        \caption{State profiles and HOCBFs for QuadProg and EIQP. $b(\bm x(t))\geq 0$ denotes safety guarantees. }
        \label{Fig3} 
\end{figure}

\section{Conclusion}
This article proposes an execution-time-certified and infeasibility-detecting solver for convex QP (including LP), which is based on an infeasible IPM algorithm. The homogeneous formulation of convex QP's KKT condition is adopted to achieve the infeasibility-detecting capability. Moreover, by exploiting this homogeneous formulation, our proposed infeasible IPM algorithm can achieve the best theoretical $O(\sqrt{n})$ iteration complexity that feasible IPM algorithms enjoy. The iteration complexity is proved to be \textit{exact} (rather than an upper bound), \textit{simple to calculate}, and \textit{data independent}, with the value $\left\lceil\frac{\log(\frac{n+1}{\epsilon})}{-\log(1-\frac{0.414213}{\sqrt{n+1}})}\right\rceil$, making it appealing to certify the execution time of online time-varying convex QPs and LPs. 

The homogeneous formulation comes at the cost of ill-conditioned linear systems to be solved; we achieved numerical robustness in our current implementation by employing LU decomposition. Future work includes: \textit{i)} developing a good preconditioner for solving the linear system at each iteration, to enable the faster decomposition to be used; \textit{ii)} taking advantage of the full Newton step feature to parallelize the predetermined sequential iterations, which can greatly improve the computational efficiency.

\section{Acknowledgement}
This research was supported by the U.S. Food and Drug Administration under the FDA BAA-22-00123 program, Award Number 75F40122C00200.

\bibliographystyle{IEEEtran}
\bibliography{ref} 

\begin{thebibliography}{10}
\providecommand{\url}[1]{#1}
\csname url@samestyle\endcsname
\providecommand{\newblock}{\relax}
\providecommand{\bibinfo}[2]{#2}
\providecommand{\BIBentrySTDinterwordspacing}{\spaceskip=0pt\relax}
\providecommand{\BIBentryALTinterwordstretchfactor}{4}
\providecommand{\BIBentryALTinterwordspacing}{\spaceskip=\fontdimen2\font plus
\BIBentryALTinterwordstretchfactor\fontdimen3\font minus \fontdimen4\font\relax}
\providecommand{\BIBforeignlanguage}[2]{{%
\expandafter\ifx\csname l@#1\endcsname\relax
\typeout{** WARNING: IEEEtran.bst: No hyphenation pattern has been}%
\typeout{** loaded for the language `#1'. Using the pattern for}%
\typeout{** the default language instead.}%
\else
\language=\csname l@#1\endcsname
\fi
#2}}
\providecommand{\BIBdecl}{\relax}
\BIBdecl

\bibitem{qin2003survey}
S.~J. Qin and T.~A. Badgwell, ``A survey of industrial model predictive control technology,'' \emph{Control engineering practice}, vol.~11, no.~7, pp. 733--764, 2003.

\bibitem{ames2016control}
A.~D. Ames, X.~Xu, J.~W. Grizzle, and P.~Tabuada, ``Control barrier function based quadratic programs for safety critical systems,'' \emph{IEEE Transactions on Automatic Control}, vol.~62, no.~8, pp. 3861--3876, 2016.

\bibitem{mellinger2011minimum}
D.~Mellinger and V.~Kumar, ``Minimum snap trajectory generation and control for quadrotors,'' in \emph{2011 IEEE international conference on robotics and automation}.\hskip 1em plus 0.5em minus 0.4em\relax IEEE, 2011, pp. 2520--2525.

\bibitem{escande2014hierarchical}
A.~Escande, N.~Mansard, and P.-B. Wieber, ``Hierarchical quadratic programming: {F}ast online humanoid-robot motion generation,'' \emph{The International Journal of Robotics Research}, vol.~33, no.~7, pp. 1006--1028, 2014.

\bibitem{bouyarmane2018quadratic}
K.~Bouyarmane, K.~Chappellet, J.~Vaillant, and A.~Kheddar, ``Quadratic programming for multirobot and task-space force control,'' \emph{IEEE Transactions on Robotics}, vol.~35, no.~1, pp. 64--77, 2018.

\bibitem{richter2011computational}
S.~Richter, C.~N. Jones, and M.~Morari, ``Computational complexity certification for real-time {MPC} with input constraints based on the fast gradient method,'' \emph{IEEE Transactions on Automatic Control}, vol.~57, no.~6, pp. 1391--1403, 2011.

\bibitem{giselsson2012execution}
P.~Giselsson, ``Execution time certification for gradient-based optimization in model predictive control,'' in \emph{Proceedings of the 51st IEEE Conference on Decision and Control}, 2012, pp. 3165--3170.

\bibitem{patrinos2013accelerated}
P.~Patrinos and A.~Bemporad, ``An accelerated dual gradient-projection algorithm for embedded linear model predictive control,'' \emph{IEEE Transactions on Automatic Control}, vol.~59, no.~1, pp. 18--33, 2013.

\bibitem{cimini2017exact}
G.~Cimini and A.~Bemporad, ``Exact complexity certification of active-set methods for quadratic programming,'' \emph{IEEE Transactions on Automatic Control}, vol.~62, no.~12, pp. 6094--6109, 2017.

\bibitem{arnstrom2019exact}
D.~Arnstr{\"o}m and D.~Axehill, ``Exact complexity certification of a standard primal active-set method for quadratic programming,'' in \emph{Proceedings of the 58th IEEE Conference on Decision and Control}, 2019, pp. 4317--4324.

\bibitem{cimini2019complexity}
G.~Cimini and A.~Bemporad, ``Complexity and convergence certification of a block principal pivoting method for box-constrained quadratic programs,'' \emph{Automatica}, vol. 100, pp. 29--37, 2019.

\bibitem{arnstrom2020complexity}
D.~Arnstr{\"o}m, A.~Bemporad, and D.~Axehill, ``Complexity certification of proximal-point methods for numerically stable quadratic programming,'' \emph{IEEE Control Systems Letters}, vol.~5, no.~4, pp. 1381--1386, 2020.

\bibitem{arnstrom2021unifying}
D.~Arnstr{\"o}m and D.~Axehill, ``A {U}nifying {C}omplexity {C}ertification {F}ramework for {A}ctive-{S}et {M}ethods for {C}onvex {Q}uadratic {P}rogramming,'' \emph{IEEE Transactions on Automatic Control}, vol.~67, no.~6, pp. 2758--2770, 2021.

\bibitem{okawa2021linear}
I.~Okawa and K.~Nonaka, ``Linear complementarity model predictive control with limited iterations for box-constrained problems,'' \emph{Automatica}, vol. 125, p. 109429, 2021.

\bibitem{wu2023direct}
L.~Wu and R.~D. Braatz, ``A {D}irect {O}ptimization {A}lgorithm for {I}nput-{C}onstrained {MPC},'' \emph{IEEE Transactions on Automatic Control}, vol.~70, no.~2, pp. 1366--1373, 2025.

\bibitem{wu2024time}
L.~Wu, K.~Ganko, and R.~D. Braatz, ``Time-certified {I}nput-constrained {NMPC} via {K}oopman operator,'' \emph{IFAC-PapersOnLine}, vol.~58, no.~18, pp. 335--340, 2024, 8th IFAC Conference on Nonlinear Model Predictive Control NMPC 2024.

\bibitem{wu2024execution}
L.~Wu, K.~Ganko, S.~Wang, and R.~D. Braatz, ``An {E}xecution-time-certified {R}iccati-based {IPM} algorithm for {RTI}-based {I}nput-constrained {NMPC},'' in \emph{63nd IEEE Conference on Decision and Control}, 2024, in press, arXiv:2402.16186.

\bibitem{wu2024parallel}
L.~Wu, L.~Zhou, and R.~D. Braatz, ``A {P}arallel {V}ector-form {$LDL^\top$} {D}ecomposition for {A}ccelerating {E}xecution-time-certified $\ell_1$-penalty {S}oft-constrained {MPC},'' \emph{arXiv preprint arXiv:2403.18235}, 2024.

\bibitem{stellato2020osqp}
B.~Stellato, G.~Banjac, P.~Goulart, A.~Bemporad, and S.~Boyd, ``{OSQP}: {A}n operator splitting solver for quadratic programs,'' \emph{Mathematical Programming Computation}, vol.~12, no.~4, pp. 637--672, 2020.

\bibitem{wu2023simple}
L.~Wu and A.~Bemporad, ``A {S}imple and {F}ast {C}oordinate-{D}escent {A}ugmented-{L}agrangian {S}olver for {M}odel {P}redictive {C}ontrol,'' \emph{IEEE Transactions on Automatic Control}, vol.~68, no.~11, pp. 6860--6866, 2023.

\bibitem{wu2023construction}
------, ``A construction-free coordinate-descent augmented-{L}agrangian method for embedded linear {MPC} based on {ARX} models,'' \emph{IFAC-PapersOnLine}, vol.~56, no.~2, pp. 9423--9428, 2023.

\bibitem{ferreau2014qpoases}
H.~Ferreau, C.~Kirches, A.~Potschka, H.~Bock, and M.~Diehl, ``qp{OASES}: {A} parametric active-set algorithm for quadratic programming,'' \emph{Mathematical Programming Computation}, vol.~6, pp. 327--363, 2014.

\bibitem{wang2009fast}
Y.~Wang and S.~Boyd, ``Fast model predictive control using online optimization,'' \emph{IEEE Transactions on Control Systems Technology}, vol.~18, no.~2, pp. 267--278, 2009.

\bibitem{gros2020linear}
S.~Gros, M.~Zanon, R.~Quirynen, A.~Bemporad, and M.~Diehl, ``From linear to nonlinear {MPC}: {B}ridging the gap via the real-time iteration,'' \emph{International Journal of Control}, vol.~93, no.~1, pp. 62--80, 2020.

\bibitem{klee1972good}
V.~Klee and G.~Minty, ``How good is the simplex algorithm,'' \emph{Inequalities}, vol.~3, no.~3, pp. 159--175, 1972.

\bibitem{mehrotra1992implementation}
S.~Mehrotra, ``On the implementation of a primal-dual interior point method,'' \emph{SIAM Journal on optimization}, vol.~2, no.~4, pp. 575--601, 1992.

\bibitem{zanelli2020forces}
A.~Zanelli, A.~Domahidi, J.~Jerez, and M.~Morari, ``{FORCES NLP}: an efficient implementation of interior-point methods for multistage nonlinear nonconvex programs,'' \emph{International Journal of Control}, vol.~93, no.~1, pp. 13--29, 2020.

\bibitem{nocedal2006numerical}
J.~Nocedal and S.~Wright, \emph{Numerical optimization}.\hskip 1em plus 0.5em minus 0.4em\relax Springer, 2006.

\bibitem{cartis2009some}
C.~Cartis, ``Some disadvantages of a {M}ehrotra-type primal-dual corrector interior point algorithm for linear programming,'' \emph{Applied Numerical Mathematics}, vol.~59, no.~5, pp. 1110--1119, 2009.

\bibitem{wright1997primal}
S.~Wright, \emph{Primal-dual interior-point methods}.\hskip 1em plus 0.5em minus 0.4em\relax SIAM, 1997.

\bibitem{ye2011interior}
Y.~Ye, \emph{Interior {P}oint {A}lgorithms: {T}heory and {A}nalysis}.\hskip 1em plus 0.5em minus 0.4em\relax New York: John Wiley \& Sons, 2011.

\bibitem{banjac2019infeasibility}
G.~Banjac, P.~Goulart, B.~Stellato, and S.~Boyd, ``Infeasibility detection in the alternating direction method of multipliers for convex optimization,'' \emph{Journal of Optimization Theory and Applications}, vol. 183, pp. 490--519, 2019.

\bibitem{bemporad2015quadratic}
A.~Bemporad, ``A quadratic programming algorithm based on nonnegative least squares with applications to embedded model predictive control,'' \emph{IEEE Transactions on Automatic Control}, vol.~61, no.~4, pp. 1111--1116, 2015.

\bibitem{arnstrom2020exact}
D.~Arnstr{\"o}m, A.~Bemporad, and D.~Axehill, ``Exact complexity certification of a nonnegative least-squares method for quadratic programming,'' \emph{IEEE Control Systems Letters}, vol.~4, no.~4, pp. 1036--1041, 2020.

\bibitem{raghunathan2021homogeneous}
A.~U. Raghunathan, ``Homogeneous formulation of convex quadratic programs for infeasibility detection,'' in \emph{2021 60th IEEE Conference on Decision and Control (CDC)}.\hskip 1em plus 0.5em minus 0.4em\relax IEEE, 2021, pp. 968--973.

\bibitem{ye1994nl}
Y.~Ye, M.~Todd, and S.~Mizuno, ``An ${O}(\sqrt{nL})$-iteration homogeneous and self-dual linear programming algorithm,'' \emph{Mathematics of operations research}, vol.~19, no.~1, pp. 53--67, 1994.

\bibitem{xu1996simplified}
X.~Xu, P.~Hung, and Y.~Ye, ``A simplified homogeneous and self-dual linear programming algorithm and its implementation,'' \emph{Annals of Operations Research}, vol.~62, no.~1, pp. 151--171, 1996.

\bibitem{ye1997homogeneous}
Y.~Ye, ``On homogeneous and self-dual algorithms for {LCP},'' \emph{Mathematical Programming}, vol.~76, no.~1, pp. 211--221, 1997.

\bibitem{andersen1999homogeneous}
E.~Andersen and Y.~Ye, ``On a homogeneous algorithm for the monotone complementarity problem.'' \emph{Mathematical Programming}, vol.~84, no.~2, pp. 375--399, 1999.

\bibitem{boyd2004convex}
S.~P. Boyd and L.~Vandenberghe, \emph{Convex optimization}.\hskip 1em plus 0.5em minus 0.4em\relax Cambridge University Press, 2004.

\bibitem{guler1993existence}
O.~G{\"u}ler, ``Existence of interior points and interior paths in nonlinear monotone complementarity problems,'' \emph{Mathematics of Operations Research}, vol.~18, no.~1, pp. 128--147, 1993.

\bibitem{monteiro1990extension}
R.~Monteiro and I.~Adler, ``An extension of {K}armarkar type algorithm to a class of convex separable programming problems with global linear rate of convergence,'' \emph{Mathematics of Operations Research}, vol.~15, no.~3, pp. 408--422, 1990.

\bibitem{laug}
E.~Anderson, Z.~Bai, C.~Bischof, S.~Blackford, J.~Demmel, J.~Dongarra, J.~Du~Croz, A.~Greenbaum, S.~Hammarling, A.~McKenney, and D.~Sorensen, \emph{{LAPACK} {U}sers' {G}uide}, 3rd~ed.\hskip 1em plus 0.5em minus 0.4em\relax Philadelphia, PA: Society for Industrial and Applied Mathematics, 1999.

\bibitem{lawson1979basic}
C.~L. Lawson, R.~J. Hanson, D.~R. Kincaid, and F.~T. Krogh, ``Basic linear algebra subprograms for {F}ortran usage,'' \emph{ACM Transactions on Mathematical Software (TOMS)}, vol.~5, no.~3, pp. 308--323, 1979.

\bibitem{jerez2011condensed}
J.~L. Jerez, E.~C. Kerrigan, and G.~A. Constantinides, ``A condensed and sparse {QP} formulation for predictive control,'' in \emph{50th IEEE Conference on Decision and Control and European Control Conference}, 2011, pp. 5217--5222.

\bibitem{KAS88}
P.~Kapasouris, M.~Athans, and G.~Stein, ``Design of feedback control systems for stable plants with saturating actuators,'' in \emph{Proceedings of the 27th IEEE Conference on Decision and Control}, 1988, pp. 469--479 vol.1.

\bibitem{bemporad1997nonlinear}
A.~Bemporad, A.~Casavola, and E.~Mosca, ``Nonlinear control of constrained linear systems via predictive reference management,'' \emph{IEEE Transactions on Automatic Control}, vol.~42, no.~3, pp. 340--349, 1997.

\bibitem{xiao2021high}
W.~Xiao and C.~Belta, ``High-order control barrier functions,'' \emph{IEEE Transactions on Automatic Control}, vol.~67, no.~7, pp. 3655--3662, 2021.

\bibitem{xiao2023safe}
W.~Xiao, C.~G. Cassandras, and C.~Belta, \emph{Safe autonomy with control barrier functions: {T}heory and applications}.\hskip 1em plus 0.5em minus 0.4em\relax Springer, 2023.

\bibitem{Khalil2002}
H.~K. Khalil, \emph{Nonlinear Systems}.\hskip 1em plus 0.5em minus 0.4em\relax Prentice Hall, third edition, 2002.

\bibitem{ames2014rapidly}
A.~D. Ames, K.~Galloway, K.~Sreenath, and J.~W. Grizzle, ``Rapidly exponentially stabilizing control lyapunov functions and hybrid zero dynamics,'' \emph{IEEE Transactions on Automatic Control}, vol.~59, no.~4, pp. 876--891, 2014.

\end{thebibliography}

\begin{IEEEbiography}[{\includegraphics[width=1in,height=1.25in,clip,keepaspectratio]{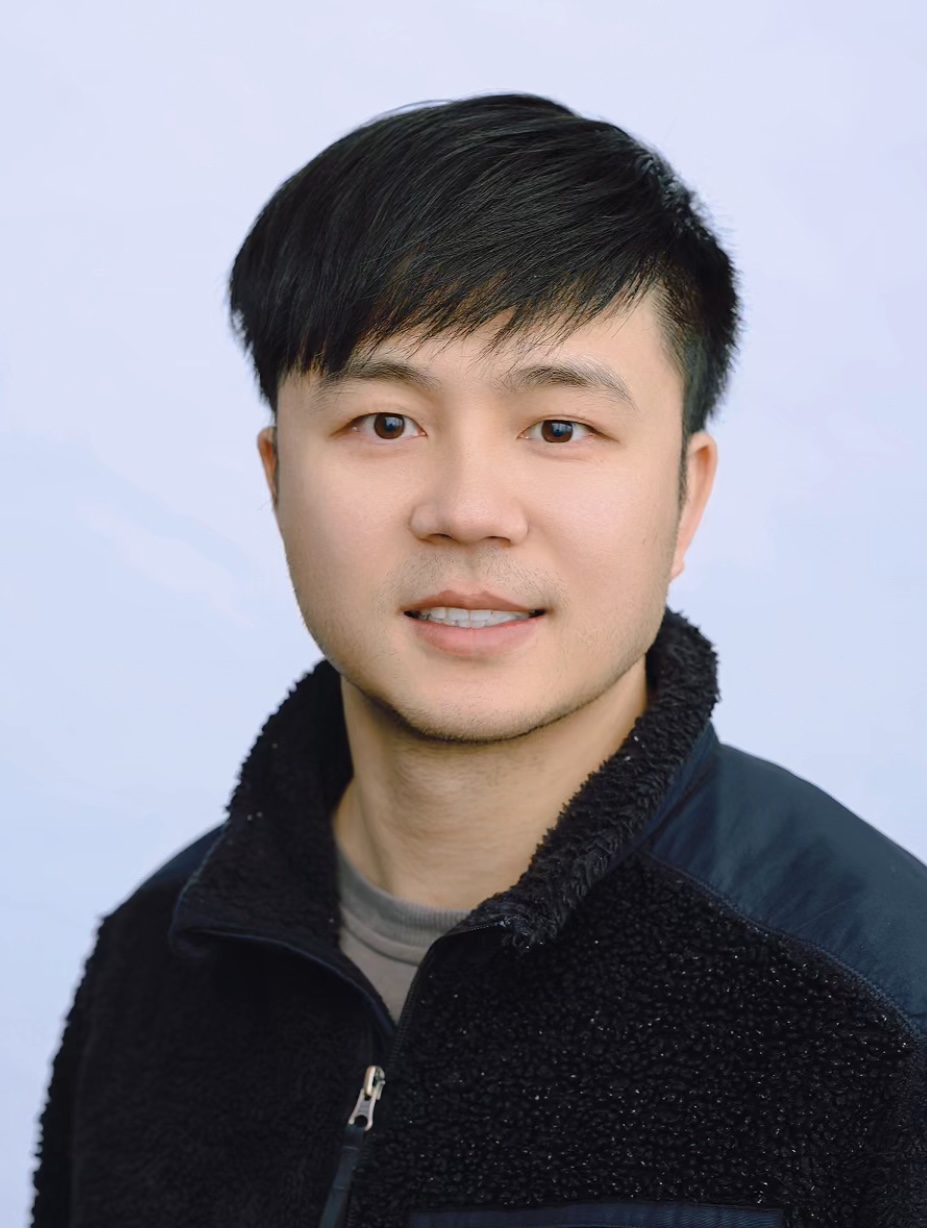}}]{Liang Wu}
received a B.E. and M.E. in chemical engineering, with majors in computational fluid dynamics and process modeling. He then turned his research interest to model predictive control and received his Ph.D. in 2023 from the IMT School for Advanced Studies Lucca supervised by Prof.\ Alberto Bemporad. 

He is now a postdoctoral associate at the Massachusetts Institute of Technology (MIT) supervised by Prof.\  Richard D. Braatz since April 2023. His research interests include quadratic programming, model predictive control, convex nonlinear programming, and control of PDE systems.
\end{IEEEbiography}

\begin{IEEEbiography}[{\includegraphics[width=1in,height=1.25in,clip,keepaspectratio]{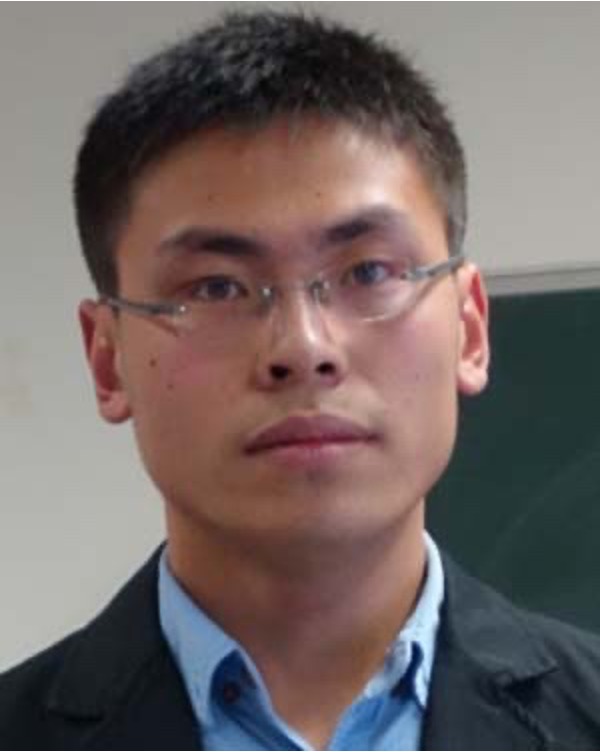}}]{Wei Xiao} (Member, IEEE) received the B.Sc. 
in mechanical engineering and automation from the
University of Science and Technology Beijing, China in 2013, the M.Sc. in robotics from the Chinese Academy of Sciences (Institute of Automation), Beijing, China in 2016, and the Ph.D. in systems engineering from Boston University, MA, USA, in  2021. He is currently a Postdoctoral Associate with the Massachusetts Institute of Technology, Cambridge, MA, USA. His current research interests include control theory and machine learning, with particular emphasis on robotics and traffic
control.

Dr. Xiao was the recipient of an Outstanding Student Paper Award at the 2020
IEEE Conference on Decision and Control.
\end{IEEEbiography}

\begin{IEEEbiography}[{\includegraphics[width=1in,height=1.25in,clip,keepaspectratio]{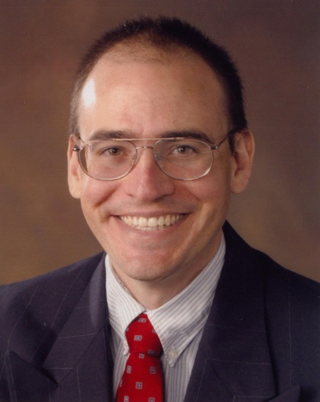}}]{Richard D. Braatz}
is the Edwin R. Gilliland Professor at the Massachusetts Institute of Technology (MIT). He received an M.S. and
Ph.D. from the California Institute of Technology. He is a past Editor-in-Chief of the IEEE Control Systems Magazine, Past President of the American Automatic Control Council, and General Co-Chair of the 2020 IEEE Conference on Decision and Control. Honors include the AACC Donald P. Eckman Award, the Antonio Ruberti Young Researcher
Prize, and best paper awards from IEEE- and IFAC-sponsored control journals. He is a Fellow of IEEE and IFAC and a member
of the U.S. National Academy of Engineering.
\end{IEEEbiography}
\end{document}